\documentclass[12pt,a4paper]{article}

\usepackage{fullpage}
\usepackage{enumerate}
\usepackage{amsmath,amssymb,amsthm} 
\usepackage{epsfig,ifpdf,graphics}
\usepackage{graphicx, color}
\usepackage[ruled,vlined,commentsnumbered,titlenotnumbered]{algorithm2e}
\usepackage[noend]{algorithmic}
\usepackage{multicol}
\usepackage{float}
\usepackage{url}
\usepackage{IEEEtrantools}

\newtheorem{theorem}{Theorem}[section]
\newtheorem{lemma}[theorem]{Lemma}

\newtheorem{definition}[theorem]{Definition}

\newcommand{\ignore}[1]{}

\let\myPushQED=\pushQED
\let\myPopQED=\popQED
\newcommand{\myignore}[1]{}
\newenvironment{proof*}
  {\let\pushQED=\myignore\begin{proof}\let\pushQED=\myPushQED}
  {\def\popQED{}\end{proof}\let\popQED=\myPopQED}

\newenvironment{description*}%
  {\vspace{-1ex}\begin{description}%
    \setlength{\itemsep}{-0.5ex}%
    \setlength{\parsep}{0pt}}%
  {\end{description}}
\newenvironment{itemize*}%
  {\vspace{-1ex}\begin{itemize}%
    \setlength{\itemsep}{-0.5ex}%
    \setlength{\parsep}{0pt}}%
  {\end{itemize}}
\newenvironment{enumerate*}%
  {\vspace{-1ex}\begin{enumerate}%
    \setlength{\itemsep}{-0.5ex}%
    \setlength{\parsep}{0pt}}%
  {\end{enumerate}}

\let\phi=\varphi

\newcommand{\keywords}[1]{\vfill\par\noindent{\textbf{\textit{Index Terms---}}#1}}

\title{\textbf{Coping with Physical Attacks on Random Network Structures}}
\author{
        Omer Gold\thanks{Department of Computer Science, Bar-Ilan University, Israel. email: omergolden@gmail.com.} 
            \and
        Reuven Cohen\thanks{Department of Mathematics, Bar-Ilan University, Israel. email: reuven@math.biu.ac.il.}
}

\date{}

\begin{document}

\pagestyle{plain}
    \bstctlcite{IEEEexample:BSTcontrol}


\begin{titlepage}
\maketitle

\begin{abstract}
Communication networks are vulnerable to natural disasters, such as
earthquakes or floods, as well as to physical attacks, such as an
Electromagnetic Pulse (EMP) attack. Such real-world events happen
at \emph{specific geographical locations} and disrupt specific
parts of the network. Therefore, the geographical layout of the network
determines the impact of such events on the network's physical topology
in terms of capacity, connectivity, and flow. 

Recent works focused on assessing
the vulnerability of a \emph{deterministic} (geographical) network to such
events. In this work, we focus on assessing the vulnerability of
(geographical)  \emph{random networks} to such disasters. 
We consider stochastic graph
models in which nodes and links are probabilistically distributed 
on a plane, and model the disaster event as a circular
cut that destroys any node or link within or intersecting the circle.

We develop algorithms for assessing the damage
of both targeted and non-targeted (random) attacks and determining
which attack locations have the expected most disruptive impact on the network. 
Then, we provide experimental results for assessing the impact
of circular disasters to communications networks in the USA, 
where the network's geographical layout
was modeled probabilistically, relying on demographic information only.
Our results demonstrates
the applicability of our algorithms to real-world scenarios.

Our novel approach allows
to examine how valuable is public information about the network's geographical area
(e.g., demography, topography, economy)
to an attacker's destruction assessment capabilities in the case the network's physical topology is hidden,
and examine the affect of hiding the actual physical location of the fibers on the attack strategy.
Thereby, our schemes can be used as a tool for policy makers and engineers to
design more robust networks by placing links along paths that avoid 
areas of high damage cuts, or identifying locations which require
additional protection efforts (e.g., equipment shielding).

Overall,
the work demonstrates that using stochastic modeling and geometry
can significantly contribute to our understanding of network survivability and resilience.

\end{abstract}


\end{titlepage}

\newpage

\setcounter{page}{2}

\sloppy

\section{Introduction}

\subsection{Background}
\label{sec:background} 

In the last decades, telecommunication networks have been increasingly crucial for 
information distribution, control of infrastructure and technological services,
as well as for economies in general. Large scale malfunctions and failures 
in these networks,  due to natural disasters, operator errors or malicious attacks
pose a considerable threat to the well being and health of individuals all over the
industrialized world. It is therefore of considerable importance to investigate the
robustness and vulnerabilities of such networks, and to find methods for
improving their resilience and stability.

The global communications infrastructure relies heavily
on physical infrastructure (such as optical fibers,
amplifiers, routers, and switches), making them vulnerable to
\emph{physical attacks, such as Electromagnetic Pulse (EMP) attacks,
as well as natural disasters, such as solar flares, earthquakes,
hurricanes, and floods} \cite{ref:webCableCut2}, \cite{OneSecondAfter}, \cite{ref:EMP-report}, \cite{Wilson04}, \cite{W-top}.
During a crisis, telecommunication is essential to facilitate the control of
physically remote agents, provide connections between emergency response personnel,
and eventually enable reconstitution of societal functions.
Such real-world disasters happen in specific geographic locations,
therefore the geographical layout of the network has a crucial factor on their impact.

Although there has been significant research on network
survivability, most previous works consider a small number of 
isolated failures or focus on shared risk groups (e.g., \cite{Bhandari99}, \cite{NMB04}, \cite{OM05}
and references therein). On the other
hand, work on large-scale attacks focused mostly on cyber-attacks
(viruses and worms) (e.g., \cite{ref:BA99}, \cite{GCABH05}, \cite{MD03}) and thereby, focus 
on the logical Internet topology.
In contrast, we consider events causing a large number of failures in a
specific geographical region, resulting in failures of network components (represented by nodes and links) 
which are geographically located within or intersecting the affected region.

This emerging field of \emph{geographically correlated failures}, 
has started gaining attention only recently, e.g., \cite{Neumayer08,NZCM.INFOCOM09,DBLP:journals/ton/NeumayerZCM11,DBLP:conf/globecom/NeumayerM11,DBLP:conf/infocom/AgarwalEGHSZ11,ieeetranAEGHSZ13, 
DBLP:conf/infocom/NeumayerM10, DBLP:conf/infocom/NeumayerEM12},  
in these works, algorithms were proposed for assessing the impact of such geographically correlated failures to a given \emph{deterministic network},
and finding a location (or set of locations) where a disaster will cause maximum disruptive damage to the network, 
measured by either capacity, connectivity or flow terms.
Various failure models were studied in these works, mainly in the form of a circular region failures or line-segment failures, such that any station or fiber within or intersecting the affected region is destroyed.
Such failures were studied under the assumption of deterministic failures,
as well as random failures (i.e. a circular region failure located randomly over the map),
and probabilistic failures (e.g. gives the ability to study cases when a component's fails in proportion to its distance from the attacked point). 
While various failure models were studied in these works, in all of them, the network's layout is assumed to be deterministic,
i.e. the geographical locations of nodes and links are known.
To the best of our knowledge, spatial non-deterministic networks (e.g. spatial random networks)
survivability were not studied under the assumptions of geographically correlated failures, and this work is the first to do so.

In this work we focus on the problem of geographically correlated failures in context of \emph{spatial random networks},
in particular, finding locations where a disaster or an attack on, will cause the most disruptive impact on the network.
Before this work, that was one of the most significant unstudied problems in the field,
as similar problems for deterministic networks was recently studied extensively.

We consider a stochastic model in which nodes
and links are probabilistically distributed geographically on a plane. 
The motivation behind it is to examine the reliability of a network where we possess only partial (probabilistic) 
information about its geographical layout. For example, a geographically hidden network where the adversary 
possesses only partial information about the network topology or no knowledge at all. 
We show that valuable probabilistic knowledge about the network's geographical layout can be modeled from publicly available data, such as
demographic maps, topographic maps, economy maps, etc.
A simple example uses the fact that in densely populated areas the probability for stations (nodes) to exist is high compared to desolated 
areas, in which it is less likely to find many stations. 
Similarly, the probability for existence of a fiber (link) between two stations can be modeled as a function of the 
distance between the stations, the population density in the station's regions, and possibly other parameters relating to the 
endpoints and geography.

\begin{figure}[t]
\centering
	  \includegraphics[width=0.9\textwidth]{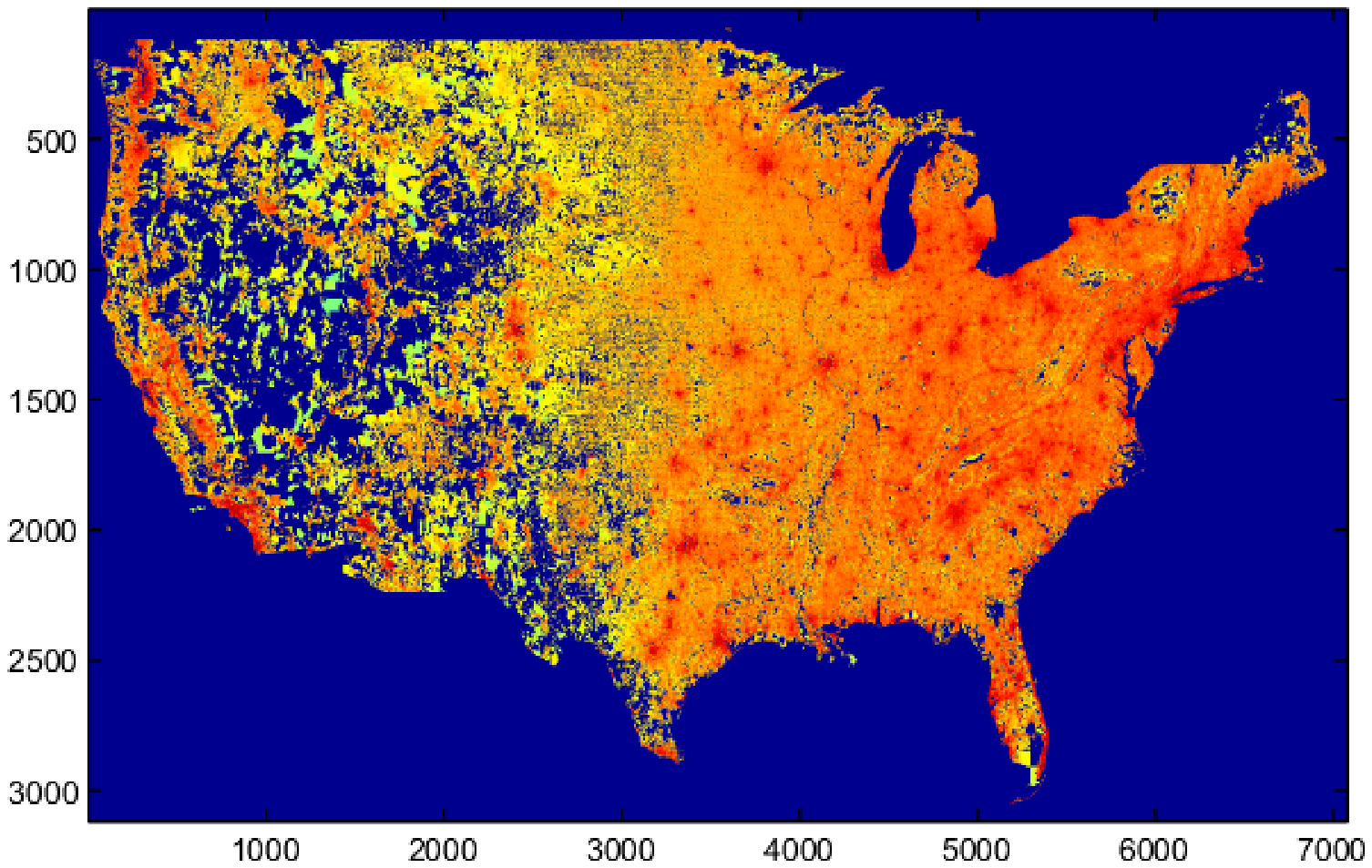}
 \caption{Color map of the the USA population density in logarithmic scale. Data is taken from \cite{columbia}.}
       \label{fig:pop_density1}
\end{figure}

We study the impact of circular geographically correlated failures centered in a specific geographical location (referred as ``circular cut").
The impact of such an event is on all the networks' component, i.e. stations, fibers (represented as nodes and links)
intersecting the circle (including its interior). Each network component located within or intersecting the circular cut
is considered to be destroyed.
The damage is measured by the total expected capacity of the intersected links,
or by the total number of failed network components.

It is relevant to note that in the case of a large-scale disaster or physical attack, where many links fail simultaneously,
current technology will not be able to handle very large-scale re-provisioning
(see for example, the CORONET project \cite{Clapp:10}). 
Therefore, we assume that lightpaths are static, implying that if a lightpath
is destroyed, all the data that it carries is lost.
However, the main data loss damage is caused by the inability to use those failed fibers, 
remaining the post-attack network (the network without the failed fibers) limited only to the components outside the attack's region.
Thus, a scenario that the location and the size of the attack causes large-scale failures is fatal for the ability to transmit data
not only from locations within the attack's region, but also from locations outside the region that are connected through fibers 
that go through the region.

The ability to probabilistically model a network using information such as demographic maps (see illustration in Fig.~\ref{fig:pop_density1}), 
terrain maps, economy maps, etc. can be used as an input to our algorithms for estimating the damage of attacks in different locations,
and determining a location where an attack will cause maximum expected disruptive damage in terms of capacity or connectivity.
This is important in assessing the expected damage from an attack by an adversary with limited knowledge.
In order to design a more robust and well defended system one can consider the resilience of
the actual network topology compared to the appropriate random model, and also consider the 
effect of hiding the actual physical location of fibers on the attack strategy and expected 
damage by an adversary.

In section~\ref{sec:numerical_results} we provide experimental results that demonstrate
the applicability of our algorithms to estimate the expected impact of circular cuts in different locations
on communication networks in the USA,
and to find locations of (approximately) worst-case cuts for this model,
where the network's layout was modeled relying on demographic information only (see Fig.~\ref{fig:pop_density1}).


Overall in this work, we study the vulnerability of various \emph{spatial random network} structures to geographically correlated failures.
That is networks which their components' locations (nodes and links) are distributed probabilistically on the plane.
Using stochastic modeling, geometric probability and numerical analysis techniques, we demonstrate a novel approach
to develop algorithms for finding locations in which circular disasters (of particular radius)
cause the expected most significant destruction, allowing to identify locations which require
additional protection efforts (e.g., equipment shielding).
We also provide an algorithm to assess the impact of a `random" circular disaster
to the random network. 
To the best of our knowledge, our work is the first to 
study such geographically correlated failures in the context of \emph{spatial random networks}.
Before this work, that was one of the most significant unstudied problems in the field,
as similar problems for deterministic networks was recently studied extensively.


\subsection{Related Work}
\label{sec:related}

The issue of network survivability and resilience has been extensively
studied in the past (e.g., \cite{MST04,Bhandari99,GR00,ZS00,MN02,NMB04,ref:CEAH00,ref:BA99} and
references therein). However, Most of these works concentrated on the logical network
topology and did not consider the physical location of nodes and links.
When the logical (i.e., IP) topology is considered,
wide-spread failures have been extensively studied \cite{GCABH05}, \cite{MD03}. 
Most of these works consider the topology of the Internet as
a random graph \cite{ref:BA99} and use percolation theory to study the 
effects of random link and node failures on these graphs. These
studies are motivated by failures of routers due to attacks by
viruses and worms rather than physical attacks.

Works that consider physical topology and fiber networks
(e.g., \cite{Crochat:2000}, \cite{NMB04}), usually focused on a small number of fiber failures 
(e.g., simultaneous failures of links sharing a common physical
resource, such as a cable, conduit, etc.). Such correlated link
failures are often addressed systematically by the concept of
shared risk link group (SRLG) \cite{IDSRLG} (see also section~\ref{sec:background}). 
Additional works explore
dependent failures, but do not specifically make use of the
causes of dependence \cite{LJ09}, \cite{S77}, \cite{TDG09}.

In contrast with these works, we focus on failures within
a specific geographical region, (e.g., failures caused by an EMP attack \cite{ref:EMP-report},\cite{Wilson04}) 
implying that the failed components do not necessarily share the same physical resource.

A closely related theoretical problem is the \emph{network inhibition problem} \cite{Phillips93}, \cite{PFL07}. 
Under that problem, each edge in the network has a destruction
cost, and a fixed budget is given to attack the network.
A feasible attack removes a subset of the edges, whose total
destruction cost is no greater than the budget.
The objective is to find an attack that minimizes the
value of a maximum flow in the graph after the attack. 
However, previous works dealing with this setting
and its variants (e.g., \cite{CSM04}, \cite{PFL07}) 
did not study the removal of (geographically) neighboring links.
Until the recent papers \cite{Neumayer08,NZCM.INFOCOM09,DBLP:journals/ton/NeumayerZCM11} by \emph{Neumayer et al.},
perhaps the closest to this concept was the problem formulated in \cite{Bienstock91}.

Similarly to deterministic networks, the subject of \emph{random networks} and \emph{survivability}
was well studied in the past (e.g. \cite{ref:BA99}, \cite{ref:CEAH00}, \cite{MD03}, \cite{CEBH01}), 
most of these works model the network as a random graph without considering the physical location of nodes and links,
and focus on the robustness of the graph's structure, ignoring a physical embedment in the \emph{plane}.
Some works considered geometry into the random model by considering distances to the graph
(e.g. \cite{DBLP:conf/podc/FabrikantLMPS03}). 
However, previous works dealing with such models did not study the removal of (geographically) neighboring links.

Recently, the subject of \emph{geographically correlated failures} was proposed
(e.g., \cite{Neumayer08,NZCM.INFOCOM09,DBLP:journals/ton/NeumayerZCM11}), 
where failures happens within a \emph{specific} geographical region and span
an extensive geographic area, disrupting all physical network equipment within the affected region. 
Novel works have been made recently to study the impact of various types of
geographically correlated failures to a given \emph{deterministic network}
\cite{DBLP:journals/ton/NeumayerZCM11}, \cite{DBLP:conf/infocom/NeumayerM10},
\cite{DBLP:conf/globecom/NeumayerM11}, \cite{MILCOM10}, \cite{MILCOM10}, \cite{ieeetranAEGHSZ13}, \cite{Agarwal:2013}. 
In these works various failure models were studied, mainly in the form of a circular region failures or line-segment failures,
such that any station or fiber within or intersecting the affected region is destroyed.
Such failures were studied under the assumptions of deterministic failures,
as well as random failures (i.e. a circular region failure located randomly over the map),
and probabilistic failures (e.g. gives the ability to study cases when a component's failure probability is proportional
to its distance from the attacked point). 
While various failure models were studied in these works, in all of them, the network's layout is assumed to be deterministic,
i.e. the geographical locations of nodes and links are known.
To the best of our knowledge, our work is the first to 
study such geographically correlated failures in the context of \emph{spatial random networks}.

We give a brief survey on works which are the most closest and relevant to our work.
Similarly to the failure model we study in our work,
most of these works consider a geographically correlated failure by a circular region of a specific radius (referred as ``circular cut"),
or by a line-segment of specific length (referred as ``line-segment cut"). Such a cut destroys any station or fiber within or intersecting it.
However, we note again, in all these works, the subject network is deterministic and given.

\begin{figure}[t]
\centering
        \includegraphics[trim = 4cm 6.5cm 2cm 5.5cm, clip, width=0.8\textwidth]{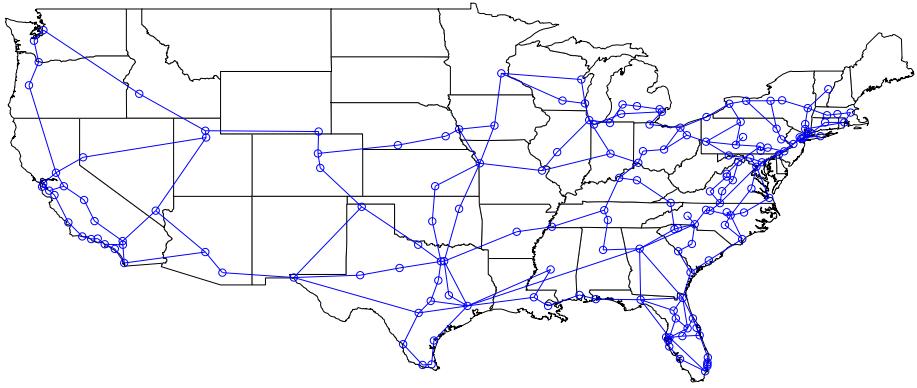}
\vspace*{-.5cm} \caption{The fiber backbone operated by a major U.S.
network provider \cite{level3} studied in \cite{Neumayer08,NZCM.INFOCOM09,DBLP:journals/ton/NeumayerZCM11}.}
       \label{fig:fiber}
\end{figure}

\paragraph{Deterministic Failures} In \cite{Neumayer08,NZCM.INFOCOM09,DBLP:journals/ton/NeumayerZCM11} \emph{Neumayer et al.} formulated the problem of finding a location where a \emph{physical} disaster or an attack on, will cause the most disruptive impact on the network in terms of capacity, connectivity and flow criterions (termed by them as the \emph{geographical network inhibition problem}).
They designed algorithms for solving the problem, and demonstrated simulation results on a U.S
infrastructure (see Fig.~\ref{fig:fiber}). To the best of our knowledge, \cite{Neumayer08,NZCM.INFOCOM09,DBLP:journals/ton/NeumayerZCM11}
were the first to study this problem.
In these works they studied the affect of circular cuts and line segments cuts under different \emph{performance measures},
in particular, provided polynomial time algorithms to find a \emph{worst-case cut}, that is a cut which maximizes/minimizes the value of the performance measure.
The performance measures of a cut, as studied in these works are:
\begin{itemize}
\item \emph{TEC} - The total expected
capacity of the intersected links with the cut.
\item \emph{ATTR} - The fraction of pairs of nodes that remain connected (also known as the average two-terminal reliability of the network)
\item \emph{MFST} - The maximum flow between a given pair of nodes $s$ and $t$.
\item \emph{AMF} - The average value of maximum flow between all pairs of nodes.
\end{itemize}
For performance measure \emph{TEC}, the worst-case
cut obtains a maximum value, while for the rest, it obtains a
minimum value.

We give illustrations from some of their simulation results on a U.S infrastructure
which are relevant to our simulation results further in this work (in section~\ref{sec:numerical_results}).
Fig.~\ref{fig:h=2.metric=CapacityCut.eps} shows that TEC value is large in areas of high link density, such as areas in
Florida, New York, and around Dallas.
Fig.~\ref{fig:r=2.metric=maxFlowLANYC} shows that cuts that minimize the MFST performance measure between
Los Angeles and NYC were found close
to both to Los Angeles and NYC, and the southwest area also
appeared to be vulnerable. These results are relevant to the simulation results for our model as shown in section~\ref{sec:numerical_results}

\begin{figure}[t]
\centering
        \includegraphics[trim = 4cm 6.5cm 2cm 5.5cm, clip, width=0.8\textwidth]{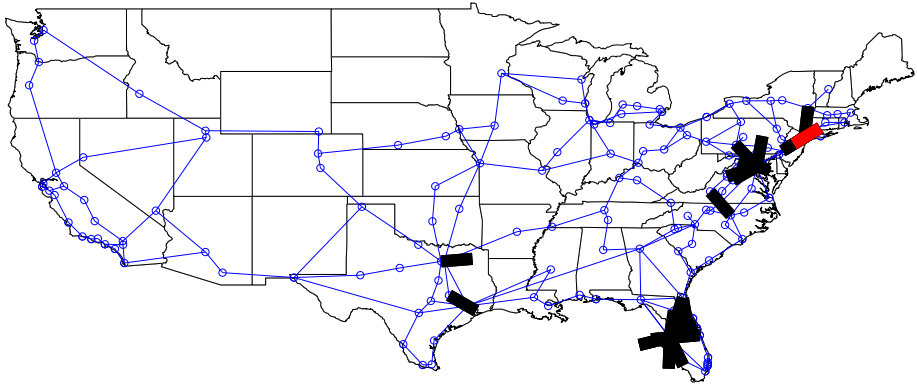}
\caption{Line segments cuts of length approximately 120 miles optimizing $TEC$ - the
red cuts maximize $TEC$ and the black lines are nearly worst-case cuts. Taken from \cite{DBLP:journals/ton/NeumayerZCM11}.}
       \label{fig:h=2.metric=CapacityCut.eps}
\end{figure}

\begin{figure}[t]
\centering
        \includegraphics[trim = 4cm 6.5cm 2cm 5.5cm, clip, width=0.8\textwidth]{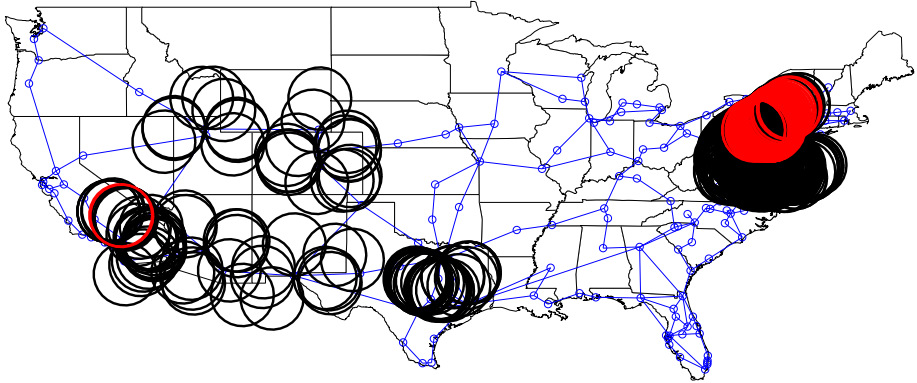}
\vspace*{-0.5cm}
\caption{The impact of circular cuts of radius approximately 120 miles on the $MFST$ between Los Angeles and NYC. Red circles represent cuts that result in $MFST=0$ and black circles result in $MFST=1$. Cuts which intersect the nodes representing Los Angeles or NYC are not shown. Taken from \cite{DBLP:journals/ton/NeumayerZCM11}.}
       \label{fig:r=2.metric=maxFlowLANYC}
\end{figure}

In \cite{MILCOM10} \emph{Agarwal et al.} presented improved runtime algorithms to the problems presented in \cite{Neumayer08,NZCM.INFOCOM09,DBLP:journals/ton/NeumayerZCM11} by the extensive use of tools from the field of \emph{computational geometry}.
In addition, they study some extensions of these problems, such as allowing \emph{multiple} disasters to happen simultaneously
and provided an approximation algorithm for the problem of finding $k$ points in which circular disasters (of particular radius) cause the most significant destruction.

\paragraph{Random Failures}
In \cite{DBLP:conf/infocom/NeumayerM10}, \cite{DBLP:conf/globecom/NeumayerM11} by \emph{Neumayer and Modiano},
they develop tools to assess the impact of a `random' geographic disaster to a given network.
They studied the impact of both random circular cut
(i.e. a circular region failure of particular radius located randomly on the plane)
and random line-segment cut.

The random location
of the disaster can model failure resulting from a natural
disaster such as a hurricane or collateral non-targeted damage
in an EMP attack.  

Intuitively, the probability of a link to fail under a random cut is proportional to the length of the link.
They also compared \emph{independent failures} versus \emph{failures correlated to a random cut} under the $ATTR$ performance measure,
assuming independent link failures such that links fail with the same
probability as in the random cut case.
Thus the probability a link fails is still a function of its length, however links
fail independently. Their results shows that correlated failures (e.g., from a random circular cut) are fundamentally different
from independent failures.

In addition they showed simple insights about network design problems in
the context of random cuts.
In the proposed problems
the location of every node is fixed and the goal is to find a set
of links most robust to some metric under some constraints.

\paragraph{Probabilistic Failures}

\begin{figure}[t]
\centering
        \includegraphics[width=0.78\textwidth]{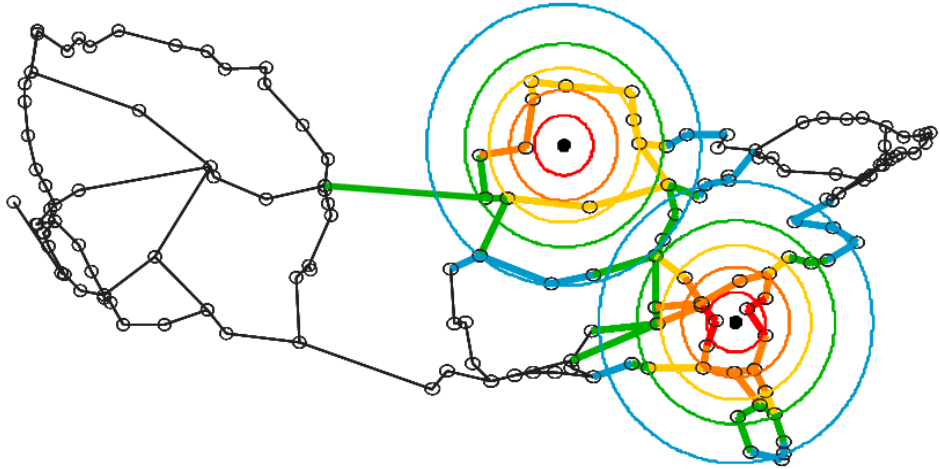}
\caption{The fiber backbone operated by a major U.S. network provider
\cite{CenturyLink} and an example of two attacks with probabilistic effects (the link colors
represent their failure probabilities). Taken from \cite{ieeetranAEGHSZ13}.}
\label{fig:probabilistic}
\end{figure}

An extensive work on the field of geographically correlated failures, generalizing previous work and providing interesting extensions
was made recently in \cite{DBLP:conf/infocom/AgarwalEGHSZ11,ieeetranAEGHSZ13} by \emph{Agarwal et al.} 
In these works they proposed a probabilistic failure model in the context of geographically correlated failures.
In this model an attack induces a spatial
probability distribution on the plane, specifying the damage
probability at each location (see Fig.~\ref{fig:probabilistic}). 
They consider probability functions
which are non-increasing functions of the distance between
the epicenter and the component, assuming these functions
have a constant description complexity. Then, they 
develop fully polynomial time approximation schemes to obtain the expected vulnerability of the
network in terms similar to the $TEC$ performance measure discussed earlier.
Their algorithms also allow the
assessment of the effects of several simultaneous events.
Another notable extension provided in these works is the study of the vulnerability of \emph{compound components},
that is a component consists by a finite number of network components, allowing \emph{lightpaths}
investigation, by noting a lightpath as a compound component comprised by a finite number of links.

Recently, \emph{Agarwal, Kaplan, and Sharir} presented in \cite{Agarwal:2013} an improved runtime algorithms
to the algorithms presented in these works, as an outcome of their result in their paper about the complexity of the union
of `random Minkowski sums' of $s-gons$ and $disks$ (where the disks' radius is a random non-negative number).
The novel techniques in these works are comprised with extensive use of tools from the field of \emph{computational geometry}.
and in particular, \emph{theory of arrangements} and \emph{randomized algorithms}.

\paragraph{Summary}
Summarizing the previous work in this emerging field of \emph{geographically correlated failures}, 
we saw that when considering the network's geographic layout to be \emph{deterministic},
various failure models were deeply studied in these works.
In our work we study the vulnerability of various \emph{spatial random network} structures.
That is networks which their components' locations (nodes and links) are probabilistically distributed on the plane.
Using stochastic modeling, geometric probability and numerical analysis techniques, we develop a novel approach
to develop an algorithm for finding locations in which circular disasters (of particular radius)
cause the expected most significant destruction. We also provide an algorithm to assess the impact of a random circular disaster
to the random network.
Before this work, that was one of the most significant unstudied problems in the field,
as similar problems for deterministic networks was recently studied extensively.
To the best of our knowledge, our work is the first to 
study such geographically correlated failures in the context of \emph{spatial random networks}.

\section{Our Stochastic Model and Problem Formulation}
\label{chap:models}
We study the model consisting of a random network 
immersed within a bounded convex set $B\subseteq\mathbb{R}^{2}$.
We consider nodes as stations and links as cables that connect stations.
Stations are represented through their coordinates in the plane $\mathbb{R}^{2}$,
and links are represented by straight line segments defined by their end-points.
The network is formed by a stochastic process in which the
location of stations is determined by a stochastic point
process. The distribution of the Poisson process is determined
by the intensity function $f(u)$ which represents the mean density of
nodes in the neighborhood of $u$. The number of nodes in a Borel
set $B$ follows a Poisson distribution with the parameter $f(B)$,
i.e., the integral over the intensity of all points in the Borel set.
Furthermore, the number of nodes in disjoint Borel sets are independent.
An introduction to Poisson Point Processes can be found in \cite{Stoyan:StochasticGeometryAndItsApplications}.

In our model we consider a network $N=((x',y')\thicksim PPP(f()),link\sim y(),c\sim H(),Rec)$
where nodes are distributed in the rectangle $Rec=[a,b]\times[c,d]$ through a \emph{Spatial
Non-Homogeneous Poisson Point Process }$PPP(f())$ where $f()$ is
the \emph{intensity function} of the $PPP$. Let $y(p_{i},p_{j})$
be the probability for the event of existence of a link between two nodes located
at $p_{i}$ and $p_{j}$ in $Rec$. $H(c,p_{i},p_{j})$ is the cumulative
distribution function of the link capacity between two connected nodes,
i.e. $P(C_{ij}<c)=H(c,p_{i},p_{j})$ where $p_{i}$ and $p_{j}$ are
the locations of nodes $i$ and $j$, respectively. It is reasonable
to assume that $y(p_{i},p_{j})$ and $H(c,p_{i},p_{j})$ can be computed
easily as a function of the distance from $p_{i}$ to $p_{j}$ and
that the possible capacity between them is bounded (denoted by $max\{capacity\}$).
We assume the following: the intensity function $f$ of the $PPP$,
$y$, and the probability density function $h$ (the derivative of $H$), 
are functions of constant description complexity. They are
continuously differentiable and Riemann-integrable over $Rec$, which also 
implies that our probability functions are of bounded variation over
$Rec$, as their derivatives receive a maximum over the compact set
$Rec\subseteq\mathbb{R}^{2}$. 

We note that Poisson process
is \emph{memoryless} and \emph{independent}, 
in particular, if we look on a specific region with high intensity
values then it is likely to have multiple stations in this region.
However, by our assumptions, it is reasonable that if a disaster destroys
some nodes in a region, then it is highly likely that it destroys other nodes located
nearby in this region. There are various point process models that can be used to model a random network,
for example, Mat\'ern hard-core process and Gibbs point processes class \cite{Stoyan:StochasticGeometryAndItsApplications},
which can be used where it is wanted that a random point is less likely to occur at the location $u$ if
there are many points in the neighborhood of $u$ or for
the hard-core process where a point $u$ is either ``permitted''
or ``not permitted'' depending whether it satisfies the hard-core
requirement (e.g. far enough from all other points). Methods similar to those presented here 
for the Poisson process can be applied to these models, as well.

\begin{definition}[Circular Cut]
A circular cut $D$, is a circle within $Rec$ determined by its center
point $p=(x_{k},y_{k})$ and by it's radius $r$, where $p \in Rec_{r}=[a+r,b-r]\times[c+r,d-r].$
\footnote{For simplicity we assume that $D$ can only appear
in whole within $Rec$.}
We sometimes denote the cut as ${\rm{cut}}(p,{r})$ and $Rec_r$ as $Rec$ (depending on the context). Such a cut destroys all fibers (links) that intersect it (including the interior of the circle).
\end{definition}

%

Our goal is to assess the vulnerability of the network to circular attacks (cuts). 
We consider a fiber to be destroyed (failed fiber) if it intersects the cut (including the interior of the circle),
namely, the attack's influence region.
The impact is measured by the \emph{total expected capacity of the intersected links} (TEC), or by the total expected number of intersected links,
which is equivalent to the previous measure when all fibers have capacity $1$ (see illustration in Fig.~\ref{fig:CircularCut}).
This will be done in three manners:

\begin{enumerate}
	\item Provide an algorithm for evaluating the total expected capacity of the intersected links (TEC) of the network with a circular attack in a \emph{specific location}.
	\item Provide an algorithm for finding an attack location (or a set of locations) which has the highest expected impact on the network, that is, a \emph{worst case attack} (one with the highest TEC value).
	\item Provide an algorithm to assess the expected impact on the network from a \emph{random circular attack}
	 (such that the attack's location is probabilistically distributed). 
\end{enumerate}

The first is done in section \ref{chap:EDCC}, the second and third in section \ref{chap:FSL}.

\begin{figure}[t]
	\vspace*{-1cm}
	\centering
		\includegraphics[width=0.5\textwidth]{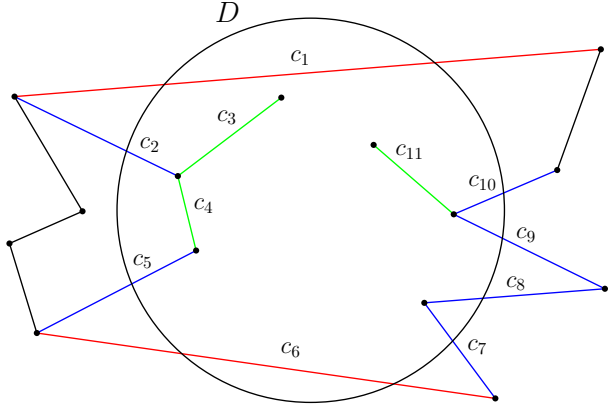}
		\caption{Example of a circular cut $D$ with TEC of $\sum\limits_{i=1}^{11} c_{i}$. The green, blue and red links depicts
		 the different intersected-link types: $\alpha-link$, $\beta-link$, and $\gamma-link$ respectively
		 (as defined in section~\ref{sec:general_idea}). The black links are not affected by $D$.}
	\label{fig:CircularCut}
\end{figure}
\section{Damage Evaluation Scheme}
\label{chap:EDCC}
\subsection{General Idea}
\label{sec:general_idea}
We develop a scheme to evaluate the total expected capacity of the intersected links (TEC) of the network with a circular attack in a \emph{specific location}. This, will be useful for the developing an algorithm that finds the attacks that have the highest expected impact on the network, as described in section~\ref{chap:FSL}.
First, we present the general idea behind it.
We divide the intersection of a cut $D$ (denoted also by ${\rm{cut}}(p,{r})$)
with a graph's edges into 3 independent types:
\begin{itemize}
\item $\alpha-link$, is the case where the entire edge is inside ${\rm{cut}}({p,{r}})$,
which means both endpoints of the edge are inside ${\rm{cut}}({p,{r}})$
(see illustration in Fig.~\ref{AlphaLink}).
\item $\beta-link$, is the case where one endpoint of the edge is inside ${\rm{cut}}({p,{r}})$
and the other endpoint is outside of ${\rm{cut}}({p,{r}})$ (see illustration in Fig.~\ref{BetaLink}).
\item $\gamma-link$, is the case where both endpoints are outside of
${\rm{cut}}({p,{r}})$ and the edge which connects the endpoints intersects
${\rm{cut}}({p,{r}})$ (see illustration in Fig.~\ref{GammaLink}).
\end{itemize}

\begin{figure}[t]

\centering
\includegraphics[scale=0.8]{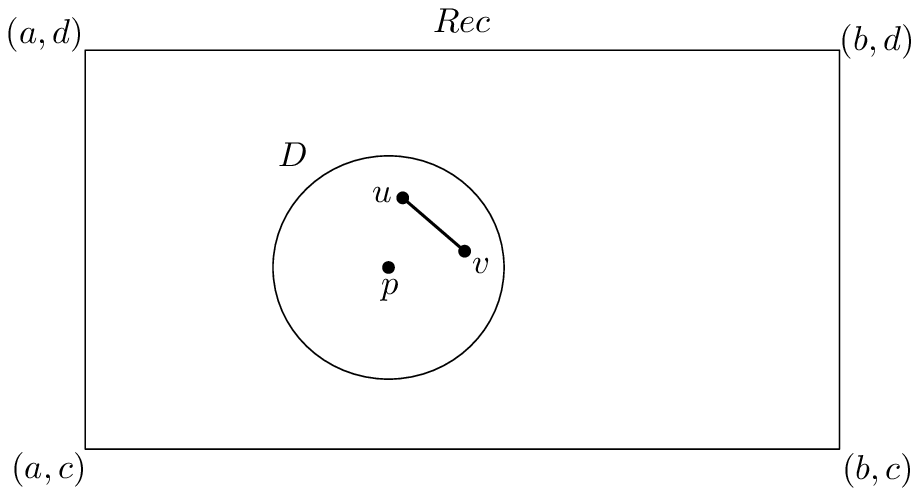}


\caption{$\alpha-link$ edge illustration}
\label{AlphaLink}


\begin{minipage}[t]{0.45\textwidth}%
\begin{center}
\includegraphics[scale=0.8]{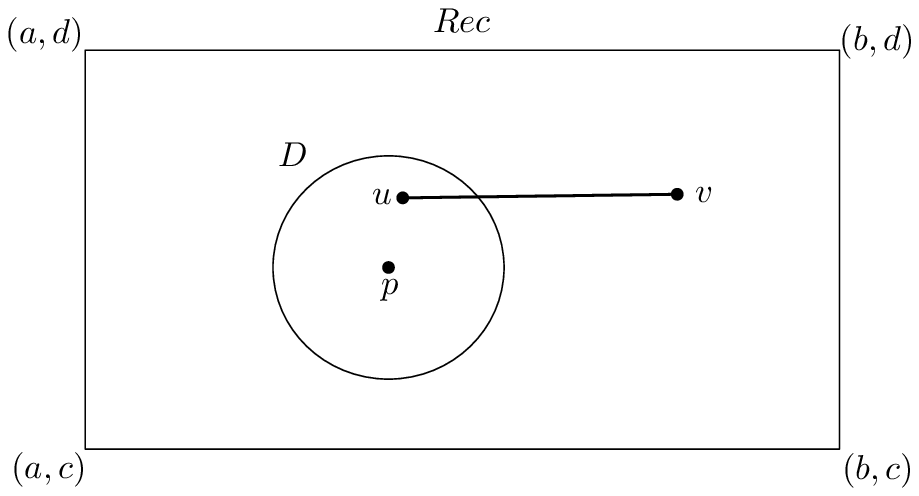}
\par\end{center}

\vspace*{-0.5cm}\caption{$\beta-link$ edge illustration}
\label{BetaLink}
\end{minipage}\hfill{}%
\begin{minipage}[t]{0.45\textwidth}%
\begin{center}
\includegraphics[scale=0.8]{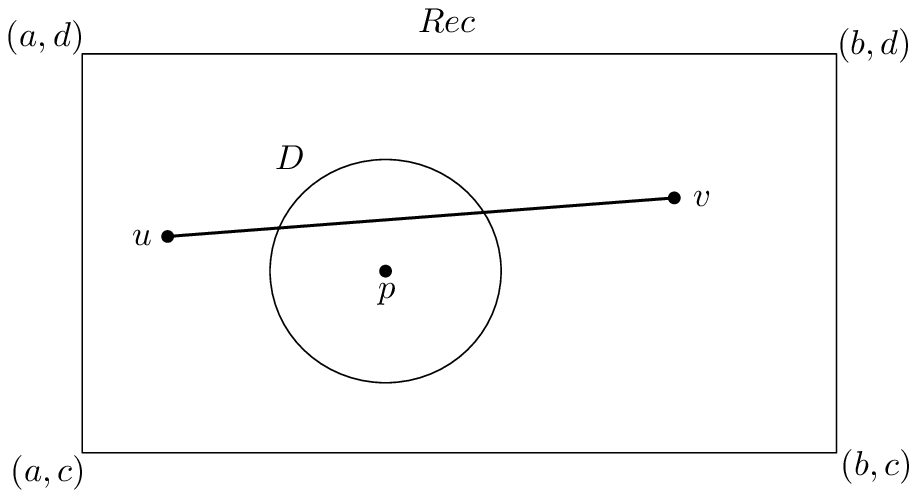}
\par\end{center}

\vspace*{-0.5cm}

\caption{$\gamma-link$ edge illustration}
\label{GammaLink}
\end{minipage}
\end{figure}

Note that any intersected link with cut $D={\rm{cut}}({p,{r}})$ belongs to exactly one of the above types.
Fig.~\ref{fig:CircularCut} depicts a circular cut with the different types of intersected links.
For $\sigma\in\{\alpha,\beta,\gamma\}$, let ${X_{\sigma}}$ be the total capacity of all the intersected $\sigma-link$ type edges with cut $D$, namely the damage determined by $\sigma-links$. Thus, it holds that the expected capacity of the intersected links of types $\alpha$, $\beta$ and $\gamma$ is determined by:
\begin{eqnarray}
E[X_{\alpha}] = \frac{1}{2}\iint\limits _{D}\iint\limits _{D}f(u)f(v)g(u,v)\,dudv\\
E[{X}_{\beta}] = \iint\limits _{Rec-D}\iint\limits _{D}f(u)f(v)g(u,v)\,dudv\\
E[{X}_{\gamma}]= \frac{1}{2}\iint\limits _{Rec-D}~\iint\limits _{Rec-D}f(u)f(v)g(u,v)\mathcal{I}(u,v,D)\,dudv\;
\end{eqnarray}
where ${g(u,v) = y(u,v)\intop_{0}^{max\{capacity\}}h(c,u,v)c\, dc}$ is the 
expected capacity between two nodes at points $u$ and $v$ 
(determined by the probability of having a link between them, times the expected capacity of this link).
$\mathcal{I}(u,v,D)$ is the indicator function, giving one if the segment $(u,v)$ intersects the circle $D$ and zero otherwise.
 
Denote by $X=X_{\alpha}+X_{\beta}+X_{\gamma}$, the total damage determined by all the intersected links with cut $D$. 
Due to the linearity of expectation, we get that the total expected capacity of the intersected links (TEC) is
 $E[X]=E[X_{\alpha}]+E[X_{\beta}]+E[X_{\gamma}]$.
Hence, it is sufficient to evaluate the expected damage caused by each of the 3 types of the intersected links separately.
Summing them all together is the total expected damage caused by $D=cut({p,{r}})$. 

%
%

\subsection{Evaluating the Damage of a Circular Cut Algorithm}
\label{sec:Evaluating}

\begin{algorithm}[t]
 \caption{EvaluateDamageCircularCut (EDCC): Approximation algorithm for evaluating the damage of a circular cut.}
\label{alg:EDCC}
\begin{algorithmic}[1]
\STATE capacity \textbf{procedure} EDCC($N,{cut}({p,{r}}),\epsilon$) 
\STATE $\displaystyle D \leftarrow Circle({\rm{cut}}({p,{r}}))$ { \footnotesize //$D$ is a representation of the circular cut.}
\STATE $ \displaystyle Grid \leftarrow ComputeGrid(Rec,r,\epsilon)$
\FOR {every $u,v\in Grid$} 
\STATE { \footnotesize // Compute the expected capacity between two points $u$ and $v$:}\\
${g(u,v) \leftarrow y(u,v)\intop_{0}^{max\{capacity\}}h(c,u,v)c\, dc}$ \\ 
\ENDFOR
\STATE  The following steps are calculated over $Grid$:
\STATE $ E[{X_{\alpha}}] \leftarrow  \frac{1}{2}\iint\limits _{D}\iint\limits _{D}f(u)f(v)g(u,v)\,dudv$
\STATE $ E[{X}_{\beta}] \leftarrow \iint\limits _{Rec-D}\iint\limits _{D}f(u)f(v)g(u,v)\,dudv$
\STATE $ E[{X}_{\gamma}]\leftarrow \frac{1}{2}\iint\limits _{Rec-D}f(u)$\textbf{evaluateGamma}$(u,D,Grid)\,du$ 
\RETURN $E[{X}_{\alpha}]+E[{X}_{\beta}]+E[X_{\gamma}]$ \\
\item [\textbf{Procedure evaluateGamma($u,D,Grid$)}]
\STATE \label{alg:Ku1} Create two tangents $t_{1}, t_{2}$ to $D$ going out from $u$. 
\STATE \label{alg:Ku2} Denote by $K_{u}$ the set of points which is bounded
by $t_{1},t_{2},D$ and the boundary of the rectangle $Rec$ (see Fig.~\ref{rec-K}). 
\RETURN $\iint\limits _{K_{u}}f(v)g(u,v)\,dv$

\end{algorithmic}
\end{algorithm}

We present an approximation algorithm $EDCC$ (see pseudo-code in Algorithm \ref{alg:EDCC})
for evaluating the total expected capacity of the intersected links (TEC) 
of the network with a circular attack (cut) $D$ in a \emph{specific location}. Later, we use this algorithm to find attack locations with the (approximately) highest  expected impact on the network.
We give two different approximation analyses for algorithm $EDCC$ output.
One is an additive approximation, and one is multiplicative. Although additive approximations in general are better than multiplicative,
our analysis of the additive approximation depends on the maximum 
value of the functions $f(\cdot)$ and $g(\cdot,\cdot)$ over $Rec$, 
where $g(\cdot,\cdot)$ stands for the expected capacity between 
two points $u,v\in Rec$. While the maximum of $f(\cdot)$ and $g(\cdot,\cdot)$ over $Rec$
can be high and affect the running time of the algorithm, 
for practical uses on ''real-life'' network it is usually low enough to make the running time reasonable.
The multiplicative analysis which does not depend on the maximum of $f(\cdot)$ and $g(\cdot,\cdot)$ over $Rec$, 
depends on the variation bound of $f(\cdot)f(\cdot)g(\cdot,\cdot)$, namely a constant $M$ which is an upper bound
for the derivative of $f(\cdot)f(\cdot)g(\cdot,\cdot)$ over $Rec$.
Define an additive $\epsilon$-approximation to the TEC as a quantity $\tilde C$ satisfying
${C}- \epsilon \leq \tilde C \leq {C}+ \epsilon$
for  $\epsilon>0$, where ${C}$ is the actual expected capacity intersecting the cut.
Similarly, define a multiplicative $\epsilon,\varepsilon$-approximation to the cut capacity  as a quantity
$\tilde C$ satisfying $(1-\varepsilon){C}-\epsilon\leq \tilde C\leq(1+\varepsilon){C}+\epsilon$.

The algorithm uses numerical integration based on the division of $Rec$ into squares of 
edge length $\Delta$ (we refer to $\Delta$ as the \emph{``grid constant''}). The different approximations are pronounced 
in the $ComputeGrid(Rec,r,\epsilon)$ function in the algorithm,
which determines the grid of constant $\Delta$.
Let $Grid$ be the set of these squares center-points. 
The algorithm evaluates the integrals numerically, using the points in $Grid$.
Intuitively: the denser is the grid, the more accurate the results, at the price 
of requiring additional time to complete.
In section \ref{sec:numerical_accuracy} we examine the relation between the accuracy parameter $\epsilon$ and the grid constant $\Delta$,
this relation determines the implementation of $ComputeGrid(Rec,r,\epsilon)$
and the running time of our algorithm.

\begin{figure}
\centering
        \includegraphics[width=0.78\textwidth]{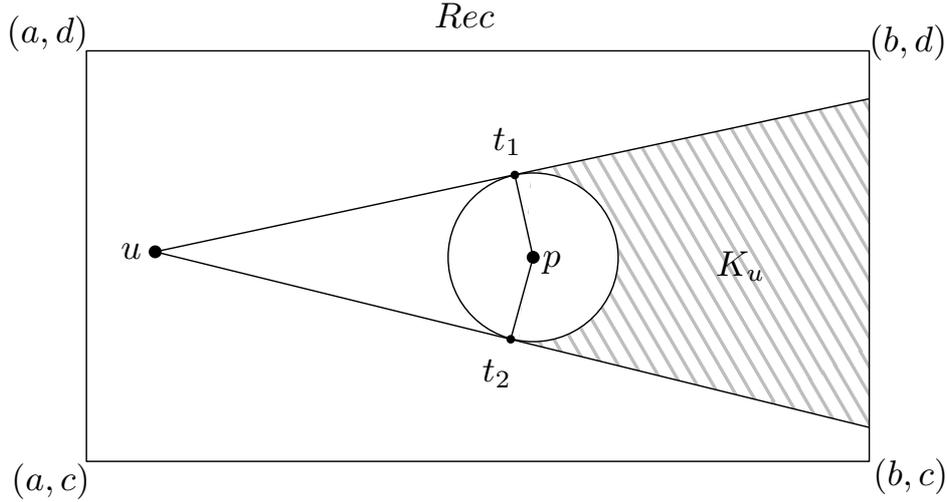}
 \caption{The area $K_u$ of the shadow by the circular cut from the point $u$.}
       \label{rec-K}
\end{figure}

When computing the expectation of the $\gamma$-links damage caused by a cut $D$, we use the following lemma:
\begin{lemma}
For every node $u\in Rec\setminus D$, the set $K_{u}$
(see lines \ref{alg:Ku1}-\ref{alg:Ku2} in Algorithm \ref{alg:EDCC} and illustration in Fig. \ref{rec-K}) satisfies: for every $v\in K_{u}$,
$(u,v)$ intersects $D$, and for every $w\in Rec \setminus (K_{u}\cup D)$,
$(u,w)$ does not intersect with $D$.
\end{lemma}
\begin{proof}
The set, $A_u$ bounded by both tangents and the boundary of the rectangle,
containing the circle, is
convex, as it is the intersection of convex sets (a triangle and a rectangle).
This set, minus the union of radii from the center of the 
circle to both tangents is a disconnected set. 
The point $u$ belongs to the connected set containing $u$, whereas the point $v$ 
belongs to the connected set $K_u$, which is is separated from $u$ by the union of radii.
That is, any path between $u$ and $v$ contained in the convex set $A_u$, 
must intersect with the union of radii, and thus with the circle.

For any point $p\in Rec$ which is not on the tangents' rays, the line segment $(u,p)$ is
either completely inside $A_u$ or completely in $Rec\setminus A_u$ (except for point $u$).
For any point $w_1\in Rec \setminus A_u$, since $A_u$ is bounded by 
the rays from $u$ which are tangent to the circle, it follows that the 
segment $(u,w_1)$ is completely outside $A_u$ (except $u$), and thus does not intersect
with the circle. 
The set $S_u = A_u \setminus (K_u \cup D)$ is a star domain, such that 
for any point $w_2 \in S_u$ the segment $(u,w_2)$ is completely contained in $S_u$,
and thus does not intersect with the circle.

\end{proof}
Thus, we run over every point $u$ in $Grid$ which is outside the cut,
and compute all possible $\gamma$-link damage emanating from $u$, using procedure $evaluateGamma(u,D,Grid)$.
Summing them all together and dividing by 2, due to double-counting, we get the total expected $\gamma$-link damage. 

\subsection{Numerical Accuracy and Running Time Analysis}
\label{sec:numerical_accuracy}
\subsubsection{Geometric Preliminaries}
\label{subsec:pre}
In this section we give theorems regarding the accuracy and the running time of Algorithm~\ref{alg:EDCC} ($EDCC$). 
The first theorem is in section~\ref{subsec:additive}, gives an additive bound on the error in the calculation of the 
damage caused by a cut of radius $r$ using numerical integration with grid constant $\Delta$. 
This gives us the relation between the accuracy parameter $\epsilon$ and the grid constant $\Delta$. This relation
gives the implementation of $ComputeGrid(Rec,r,\epsilon)$ in $O(1)$
by choosing $\Delta$ small enough such that the error will be not more than $\epsilon$. 
Then, in Theorem \ref{thm:runtime_add} we give a bound on the running time of the algorithm for any accuracy parameter $\epsilon>0$.
 
In section~\ref{subsec:multiplicative}, we give a combined multiplicative and additive bound on the error for a grid constant $\Delta$,
which is different from the additive bound by being independent on the maxima of the functions $f$ and $g$,
thus, for a given accuracy parameters $\varepsilon, \epsilon$, one can choose $\Delta$ which is independent on the maxima of the functions $f$ and $g$, small enough such that the error will be not more than $(1+\varepsilon)C +\epsilon$, where $C$ is the actual TEC of the cut. 
With this multiplicative-additive approximation we can bound the running time independently on the maxima of the functions $f$ and $g$,
useful for the case these maxima are high (usually in "real-life" networks it is not the case and the additive approximation will have a reasonable running time, as described in section \ref{sec:Evaluating}).

We restrict our results to the case where $\Delta<r/2$, as otherwise the approximation is too crude
to consider.

\begin{figure}
\centering
        \includegraphics[width=0.6\textwidth]{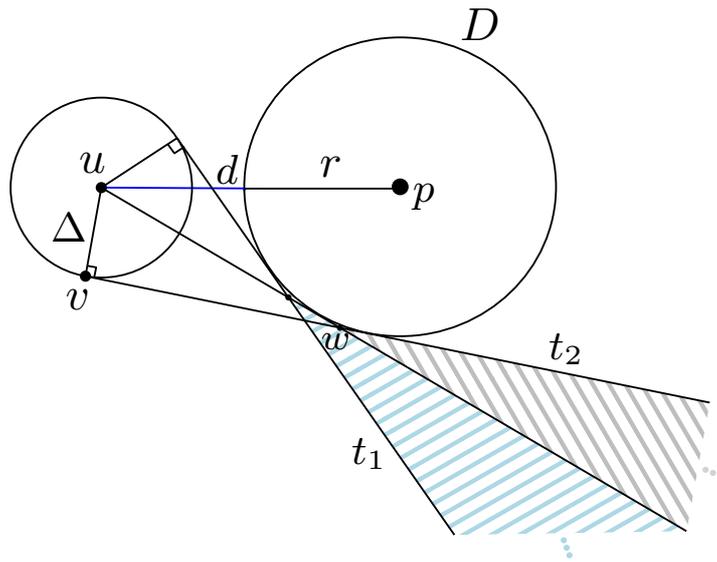} 
 \caption{$t_1$ and $t_2$ are tangents to a circle of radius $\Delta$ centered at $u$ and to the circular cut.
 The colored areas depicts the extremum of difference for possible $\gamma$-links endpoints emanating from a point within the circle centered at $u$ at one side of the cut.}
       \label{fig:approx3}
\end{figure}

\begin{figure}
\centering
        \includegraphics[width=0.6\textwidth]{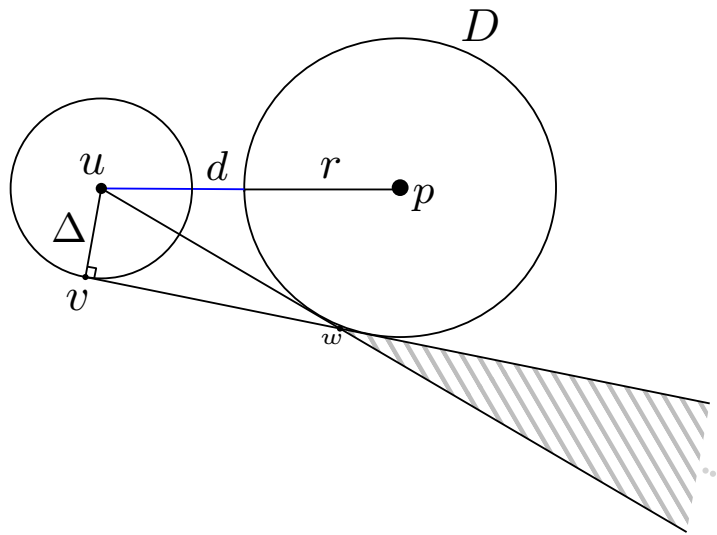} 
 \caption{The area in grey depicts the extremum of difference for possible $\gamma$-links endpoints emanating from a point within the circle centered at $u$ at one side of the cut.}
       \label{fig:approx1}
\end{figure}

\begin{figure}
\centering
        \includegraphics[width=0.38\textwidth]{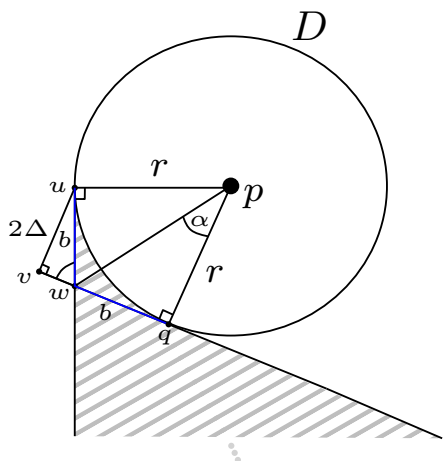} 
 \caption{The perpendicular tangents from a point $u$ on the circular cut and from a point $v$ at distance $2\Delta$ from the circular cut.
 The area in grey depicts the extremum of difference for possible $\gamma$-links endpoints emanating from $u$ and $v$ at one side of the cut.}
       \label{fig:2delta}
\end{figure}

Some technical results are needed for proving the theorems in this section.
We first notice that, given a point $u$ at a distance $d$ from a circular cut,
and a point $v$ at a distance $\Delta$ from $u$, the  difference between
the area of endpoints of $\gamma$-links starting at point $u$ to those starting
at point $v$ is bounded in the area between 
(a) the segment from $u$ to the boundary of 
$Rec$ tangent to the cut
(b) the segment from $v$ to the boundary of 
$Rec$ tangent to the cut
(c) the boundary of $Rec$.
Plus the area bounded by (a), (b) and the boundary of the circular cut.
Furthermore, the extrema of the area of difference are obtained in the cases where 
the segment from $v$ to the boundary of 
$Rec$ tangent to the cut is also tangent to the circle centered at $u$ with radius $\Delta$. See Fig. \ref{fig:approx3} and \ref{fig:approx1}.

For an accurate numeric bound of the error in the theorems of this section, one should add a factor of $2$ to the results of the following lemmas
(to bound also the area of possible $\gamma$-links endpoints obtained by the additional tangents emanating from $u$ and $v$ to the other side of the circular cut).

\begin{lemma}
\label{lem:gamma_error}
The area bounded by (a) the tangent at a point $u$ on the circumference of 
the circular cut (b) the tangent to the cut from any point $v$ at distance $2\Delta$ from $u$ 
(c) the circumference of the cut and (d) the boundary of $Rec$, is
bounded by $c\sqrt{\Delta}$ for some constant $c$. See Fig. \ref{fig:2delta}.
\end{lemma}
\begin{proof}
The extremum of the difference between the two tangents is when the
tangent to the circle is also a tangent to the circle of radius $2\Delta$ centered at $u$,
i.e., when the segment of length  $2\Delta$ starting at point $u$ is perpendicular
to the tangent at its other endpoint, $v$. See Fig. \ref{fig:2delta}.

Let $q$ be the point of intersection of the second tangent and the circular cut,
and let $w$ be the point of intersection between the two tangents. Let $p$ be the
center of the circle. We have $\angle u p w=\angle w p q=\alpha$
and angle $\angle u w v=2\alpha$. 

Let $b=d(u,w)=d(w,q)$ and $d(w,v)=a=\sqrt{b^2-4\Delta^2}$. We have $\tan\alpha=b/r$,
where $r$ is the radius of the circle,
and $\tan2\alpha=\dfrac{2\Delta}{\sqrt{b^2-4\Delta^2}}$. Using $\tan2\alpha=2\tan\alpha/(1-\tan^2\alpha)$
one obtains $b^2=(\Delta^2r^2+\Delta r^3)/(r^2-\Delta^2).$ Thus, $ {r\Delta} \leq b^2 \leq {2r\Delta} $. 
Assume that both tangents hit the same side of $Rec$, the area of the triangle bounded by both tangents and the boundary of $Rec$ is
bounded by 
\begin{equation} \label{eq:sameside}
\frac{1}{2} \mathcal{D}^2\sin2\alpha\leq \mathcal{D}^2\sin\alpha\leq \mathcal{D}^2\tan\alpha\leq \mathcal{D}^2\sqrt\frac{2\Delta}{r},
\end{equation}
where $\mathcal{D}$ is the diagonal of $Rec$.
Now, if the tangents hit different sides of $Rec$, there are two cases:

If they hit perpendicular sides, then the area bounded by them and the sides is a quadrangle, taking a line segment from the corner point of $Rec$ which is the intersection of the perpendicular sides to $w$ we triangulate this area, obtaining two triangles.
The angle, near point $w$ of each triangle is bounded by $2\alpha$ which is bounded by $\pi/2$, since we assume $\Delta < r/2$.
Thus from (\ref{eq:sameside}) the area is bounded by $2\mathcal{D}^2\sqrt{\dfrac{2\Delta}{r}}$.

If the tangents hit parallel sides we can triangulate the area bounded by them and the sides, by taking two line segments from 
the corners of $Rec$ within this area to the point $w$, obtaining three triangles. 
Using similar arguments as the previous case, we obtain that the area is bounded by $3\mathcal{D}^2\sqrt{\dfrac{2\Delta}{r}}$.

Finally, the area bounded by the segments $uw$ and $wq$ and the circular arc $uq$ is bounded 
by the area of triangle $\triangle uwq$, which, in turn is bounded by $b^2\leq 2r\Delta\leq \mathcal{D}^2\sqrt{\dfrac{2\Delta}{r}} $.
\end{proof}

\begin{lemma}
\label{lem:gamma_error_gen}
For a point $u$ at a distance $d$ from a circular cut $\mathrm{cut}(p,r)$ and a given $\Delta>0$,
the area between the circumference of the cut, the boundary of $Rec$,  the tangent to the cut from $u$ and the tangent to the cut  from a
point $v$ located at a distance $\Delta$ from $u$ is bounded by $c\sqrt{\Delta}$ for some constant $c$. See Fig. \ref{fig:approx1}.
\end{lemma}
\begin{proof}
Let $w$ be the point of intersection between the two tangents (from point $u$ to the cut 
and from point $v$ to the cut). We have $\sin(\angle uwv)={\Delta}/(\sqrt{(r+d)^2-r^2}+b)$, where $b$ is the distance 
between the point of intersection of each of the tangents with the cut and the point $w$.
we have $\sin(\angle uwv)={\Delta}/(\sqrt{2rd+d^2}+b)$. 

Now for $d<\Delta$ the point $v$ is located within a distance of $d+\Delta<2\Delta$ from the circumference of the 
cut, and the Lemma follows from Lemma \ref{lem:gamma_error}.

If $d>\Delta$ we have 
$\sin(\angle uwv)<{\Delta}/(\sqrt{2rd})<{\Delta}/(\sqrt{2r\Delta})$. Thus, if the tangents hit the same side of $Rec$ the triangular area between the tangents and the side is bounded by $\mathcal{D}^2\sqrt{\Delta/r}$,
where $\mathcal{D}$ is the diagonal of $Rec$.

If the tangents hit different sides, then we can triangulate the area bounded by them and the sides similarly as in Lemma \ref{lem:gamma_error}.
Denote by $\alpha=\angle uwv$ the angle between the two tangents.
For each triangle, the angle near point $w$ is bounded by $\alpha$ which is bounded by $\pi/2$ (since we assume $\Delta<r/2$).
Thus, the area bounded by the tangents and $Rec$ (which is not containing the circle) is bounded by
$3\mathcal{D}^2\sqrt{\Delta/r}$.

Similarly to Lemma \ref{lem:gamma_error} the area between the two tangents and the circle
is also bounded by $b^2\sin(\pi-\angle uwv)=b^2\sin(\angle uwv)<\sqrt{2r^3\Delta}$, as $b< r$ since both tangents intersect
the same quadrant of the cut.
\end{proof}

\subsubsection{Additive Approximation}
\label{subsec:additive}

\begin{theorem}
\label{thm:acc_add}
For a grid of constant $\Delta$, a point $p \in Rec$, and the result $\tilde C$ for $cut(p,r)$ obtained by Algorithm \ref{alg:EDCC}, 
it holds that $C-\epsilon<\tilde C<C+\epsilon$,
where $C$ is the actual TEC value for $cut(p,r)$, and $\epsilon=c_0\cdot\sqrt{\Delta}$ for some constant $c_0>0$ 
that depends on the maximum values of $f(\cdot)$ and $g(\cdot,\cdot)$, their variation bound, the sides of $Rec$ and the radius $r$.

\end{theorem}
\begin{proof} 
By standard arguments on numerical integration the error in calculating the integral 
over any region is bounded by $M\Delta$ times the area of integration (that is bounded by
the area of $Rec$, $|Rec|$), where $M$ is a bound on the variation rate for $f(\cdot)$, $g(\cdot)$,
and the product $f(\cdot)f(\cdot)g(\cdot,\cdot)$ over $Rec$.
Additionally, the cumulative error value $|C-\tilde C|$ consists of the following:

Any point in a grid square is within a distance of $\Delta/\sqrt{2}<\Delta$ of the grid point (square center).
The additional difference in the integral over $\alpha$ and $\beta$ links is bounded by the integral over the area of inaccuracy 
around the circular cut (grid squares which are partially in the cut and partially outside). This is bounded by 
an annulus of radii $\left[ r-\Delta/\sqrt{2},r+\Delta/\sqrt{2} \right]$ around the center of the circular cut of area  $2\sqrt{2}\pi r\Delta$. Thus, we obtain an error which is bounded by $2\sqrt{2}\pi r\Delta T |Rec|$,
where $T$ is a bound on the maximum value of $f(\cdot)$, $g(\cdot)$, and the product $f(\cdot)f(\cdot)g(\cdot,\cdot)$ over $Rec$.

The additional error is in the calculation of $\gamma$-links, and obtained in three terms, one term is determined 
in the procedure $evaluateGamma(u,D,Grid)$ where the area of inaccuracy is around $K_u$ (grid squares which are partially in $K_u$ and partially outside). This area is bounded by $2\sqrt{2}\mathcal{D}\Delta$ where $\mathcal{D}$ is the diagonal length of $Rec$. This gives an error term which is bounded by $2\sqrt{2}\mathcal{D}\Delta T|Rec|$.

The second error term in the $\gamma$-links calculations is obtained by considering the change in the functions $f(u)$ w.r.t $f(w)$, 
and $g(u,\cdot)$ w.r.t $g(w,\cdot)$ in $K_u \cap K_w$, for a point $w$ within distance $\Delta$ from $u$.
For convenience, for a point $u$, denote $\tilde f_u(v) = f(u)f(v)g(u,v)$.
Taking into account the change in the integrated function $\tilde f_u(v)$ and $f_w(v)$
over $K_u \cap K_w$, for a point $w$ within distance $\Delta$ from $u$,
we obtain an error, bounded by
$\iint\limits _{K_u \cap K_w} \left(\tilde f_u(v) + M\Delta \right) \,dv - \iint\limits _{K_u \cap K_w} \tilde f_u(v)\, dv
\leq \iint\limits _{K_u \cap K_w}  M\Delta \,dv \leq M\Delta|Rec|$.
Thus, obtaining an error bounded by $\Delta M|Rec|^2$. 

The third, and the most significant error term, is obtained using Lemma \ref{lem:gamma_error_gen} which implies that for any point
$u\in Rec\setminus D$ in a grid square (except grid squares which are partially in $D$ and partially outside), the area of symmetric difference between $K_u$ and $K_w$ for the grid point (square center) $w$ nearest to $u$
(such that the euclidean distance $d(u,w)<\Delta$)
is bounded by
$a\sqrt{\Delta}$, where $a$ is some constant (see section~\ref{subsec:pre}), 
depending on the radius of the cut $r$ and the sides of the rectangle $Rec$. 
Thus, we obtain an error bounded by $a|Rec|T\sqrt\Delta$.

Taking into account the errors in this numerical integration from all terms above, one obtains 
that the leading term in the error, as $\Delta \rightarrow 0$, is $|C-\tilde C|\leq const\mathcal{D}^2|Rec|T\sqrt{\dfrac{\Delta}{r}}$,
where $\mathcal{D}$ is the diagonal length of $Rec$ (see section~\ref{subsec:pre}).
Thus, the accuracy
depends on $\sqrt\Delta$, as well as on the sides of the rectangle, the radius of the cut, the maxima 
of $f(\cdot)$, $g(\cdot,\cdot)$ and their variation bound in $Rec$.
\end{proof}
Using Theorem \ref{thm:acc_add} and the given in section~\ref{subsec:pre}, the function
$ComputeGrid(Rec,r,\epsilon)$ in the algorithm can 
be implemented by selecting the value of $\Delta$ guaranteeing that the additive error will be at most $\epsilon$.

We now give a bound on the running time of the $EDCC$ algorithm for any additive accuracy parameter $\epsilon>0$.
\begin{theorem}
\label{thm:runtime_add}
For a $Rec$ of area $A$ with diagonal length $\mathcal{D}$, attack of radius $r$, and an additive accuracy parameter $\epsilon>0$, 
the total running time of Algorithm \ref{alg:EDCC} ($EDCC$) is
$O\left(\dfrac{A^{10}\mathcal{D}^{16}T^8}{\epsilon^8r^4} + \dfrac{A^{10}T^4M^4}{\epsilon^4} \right)$,
where $T$ is a bound on the maximum value of $f(\cdot)$, $g(\cdot)$, and the product $f(\cdot)f(\cdot)g(\cdot,\cdot)$ over $Rec$,
and $M$ is a bound on the variation rate for these functions over $Rec$.
\end{theorem}

\begin{proof}
The algorithm is based on performing numerical integration over pairs of grid points (square-center points). The denser is the grid, the more pairs of points we have in the grid, thus the running time is determined by the grid constant $\Delta$ (the squares' side length) which is set by $ComputeGrid(Rec,r,\epsilon)$ function at the beginning of the algorithm. The number of grid points in the rectangle is $A/\Delta^2$. 
Thus the running time is at most proportional to
the number of pairs of grid points, which is $O\left(A^2/\Delta^4\right)$.

Now, from the given in section~\ref{subsec:pre} and the proof of Theorem~\ref{thm:acc_add} we obtain that as $\Delta \rightarrow 0$
$$\epsilon=O\left(\mathcal{D}^2 AT\sqrt{\frac{\Delta}{r}} + \Delta MA^2 \right)$$ 
Thus, when $M=O\left(\dfrac{\mathcal{D}^4T^2}{\epsilon r}\right)$ (reasonable in practical usages),
$\Delta = \Omega\left(\dfrac{\epsilon^2 r}{\mathcal{D}^4A^2T^2}\right)$,
and the total running time of the algorithm is 
$$O\left(A^2/\Delta^4\right) = O\left(\frac{A^{10}\mathcal{D}^{16}T^8}{\epsilon^8r^4}\right).$$
Otherwise, if $M=\omega\left(\dfrac{\mathcal{D}^4T^2}{\epsilon r}\right)$, 
we obtain that $\Delta=\Omega\left(\dfrac{\epsilon}{MA^2}\right)$,
and the running time of the algorithm is
$$O\left(\dfrac{A^{10}M^4}{\epsilon^4}\right).$$
\end{proof}

\subsubsection{Multiplicative Approximation}
\label{subsec:multiplicative}
Since the constant in Theorem \ref{thm:acc_add} depends on the maximum value of the functions $f$ and $g$, which may be undesirable 
in case these maxima are high, we have the following theorem, giving a combined  additive and multiplicative accuracy
with the constants independent of the maxima of $f$ and $g$. Using the following Theorem \ref{thm:acc_mult}, the function $ComputeGrid(Rec,r,\epsilon)$ can be modified to
a new function $ComputeGrid(Rec,r,\epsilon,\varepsilon)$
which can be implemented by selecting the value of $\Delta$ guaranteeing that the additive error will be at most $\epsilon$
and the multiplicative error will be at most $\varepsilon$, as described in the following theorem.

\begin{theorem}
\label{thm:acc_mult}
For a grid of constant $\Delta$, a point $p \in Rec$, and the result $\tilde C$ for $D=cut(p,r)$ obtained by Algorithm~\ref{alg:EDCC}, 
it holds that $(1-\varepsilon)C-\epsilon<\tilde C<(1+\varepsilon)C+\epsilon$,
where $C$ is the actual TEC value for $cut(p,r)$, for $\epsilon=c_1\cdot\sqrt{\Delta}$, $\varepsilon=c_2\cdot\sqrt{\Delta}$,
such that $c_1$ and $c_2$ depend only on $Rec$, $r$, and $M$ the bound on the variation
of $f(\cdot)$, $g(\cdot)$, and $f(\cdot)f(\cdot)g(\cdot,\cdot)$ over $Rec$, but are independent on their maximum values.
\end{theorem}
\begin{proof} 
From the proof of Theorem \ref{thm:acc_add}, the standard error in the numerical integration over the grid depends
only on the grid constant $\Delta$, the radius of the cut $r$, and the bounded variation rate $M$ of the integrated function.

For a point $u\in Rec\setminus D$ and the closest grid point $w\in Rec\setminus D$ nearest to $u$ (with a distance at most $\Delta$ from each other), denote by $R'$ the segment of the tangent going out from $u$ to one side of $D$ within $Rec$,
similarly, denote by $R$ the segment of the tangent going out from $w$ to the same side of $D$ within $Rec$.
Denote by $R(x)$ the point on $R$ with coordinate $x$, and similarly for $R'(x)$,
where the $x$-axis is taken to be the line that does not intersect with $D$ and going through
the angle bisector for the angle between $R$ and $R'$ \footnote{There are two pairs of vertical angels by formed by the intersection of $R$ with $R'$ and thus, two angle bisectors, the angle bisector of one pair is intersecting the cut and the other does not intersecting the cut, 
we refer to the pair of which their angle bisector does not intersect the cut.}.
By the proof of Lemma~\ref{lem:gamma_error_gen} we obtain that the euclidean distance $||R'(x)-R(x)||$ satisfies
$||R'(x)-R(x)|| < a\sqrt{\Delta}$ for any $x\in Rec$, where $a$ is a constant depending on the sides of $Rec$ (see section~\ref{subsec:pre}),
this is obtained directly from the proof of Lemma~\ref{lem:gamma_error_gen} by using similarity of triangles properties.
For convenience, for a point $u$, denote $\tilde f_u(v) = f(u)f(v)g(u,v)$, and write it in Cartesian coordinates $f_u(x,y)$ 
(with respect to the axis described above).
The TEC from $\beta-links$ and $\gamma-links$ going out from $w$ is bounded by
$$\iint\limits _{{K_u}\cup D} \tilde f_u(x,y)\,dxdy + \int dx \int\limits_{R(x)-a\sqrt{\Delta}}^{R(x)} \left(\tilde f_u(x,y) +M(a\sqrt{\Delta} + \Delta)\right)\,dy,$$
thus, the difference between the TEC values of $\beta-links$ and $\gamma-links$ emanating from $u$
to $\beta-links$ and $\gamma-links$ emanating from $w$ is bounded by

\begin{equation}
\label{eq:mult1}
\int \int\limits_{R(x)-a\sqrt{\Delta}}^{R(x)} \left(\tilde f_u(v) + M\sqrt{\Delta}(a + \sqrt{\Delta})\right)\,dxdy 
\end{equation}
Taking strips of length $a\sqrt{\Delta}$ we obtain the following:
\small
\begin{align}
\label{eq:seq}
\int \int\limits_{R(x)-a\sqrt{\Delta}}^{R(x)} \left(\tilde f_u(v) + M\sqrt{\Delta}(a + \sqrt\Delta) \right) \,dxdy  
&\leq \int \int\limits_{R(x)-2a\sqrt{\Delta}}^{R(x)-a\sqrt{\Delta}} \left( \tilde f_u(v) + M\sqrt{\Delta}(2a + \sqrt\Delta) \right) \,dxdy \nonumber \\
\leq \int \int\limits_{R(x)-3a\sqrt{\Delta}}^{R(x)-2a\sqrt{\Delta}} \left( \tilde f_u(v) + M\sqrt{\Delta}(3a + \sqrt\Delta) \right) \,dxdy 
&\leq \dotsb 
\end{align}
\normalsize
Note that 
\begin{equation}
\label{eq:mult2}
\iint\limits _{{K_u}\cup D} \tilde f_u(x,y)\,dxdy \geq \int \int\limits_{R(x)-2r}^{R(x)} \tilde f_u(v)\,dxdy
\end{equation}
Thus, when integrating over $K_u \cup D$ by summing integrations of strips with length $a\sqrt{\Delta}$,
at least $\dfrac{2r}{a\sqrt{\Delta}}$ such strips are needed.

Now taking the average of the sequence (\ref{eq:seq}) of length $\dfrac{2r}{a\sqrt{\Delta}}$ (note that we allow fraction of an element, e.g., the sequence can be a less than one long) and from (\ref{eq:mult2}) we obtain that
the difference between the TEC values (given in \ref{eq:mult1}) is bounded by
\small
\begin{align}
&\frac{a\sqrt{\Delta}}{2r} \left[ \iint\limits _{{K_u}\cup D} \tilde f_u(x,y)\,dxdy \right.
+\left. \int \int\limits_{R(x)-a\sqrt{\Delta}}^{R(x)} M\sqrt{\Delta}(a + \sqrt\Delta) \,dxdy \right. \nonumber \\
&+\left. \dotsb + \int \int\limits_{R(x)-2r}^{R(x)} M\sqrt{\Delta}(\frac{2r}{a\sqrt{\Delta}}a + \sqrt{\Delta}) \,dxdy \right] \nonumber \\
&\leq \frac{a\sqrt{\Delta}}{2r} \iint\limits _{{K_u}\cup D} \tilde f_u(x,y)\,dxdy 
+ aM\mathcal{D}\sqrt{\Delta}(r + a\sqrt{\Delta} + \Delta),
\end{align}
\normalsize
where $\mathcal{D}$ is the diagonal length of $Rec$.

Thus, in leading terms, as $\Delta \rightarrow 0$, we obtain an error with multiplicative factor $\dfrac{a}{2r} \sqrt{\Delta}$
and an additive term of $aMr\mathcal{D}\sqrt{\Delta}$ between the two TEC values. 
Using similar techniques, one can bound the error obtained by the areas of inaccuracy:

(i) Grid squares that are partially in cut $D$ and partially outside.
This error can be bounded using similar methods when representing
the integrated function in polar coordinates with respect to the center of $D$.
Since this area is bounded by an annulus of radii $\left[ r-\Delta/\sqrt{2},r+\Delta/\sqrt{2} \right]$ around the center of $D$, 
we can take strips of length $\Delta$. Thus,
this can be bounded by a multiplicative factor of $b\Delta$ and additive term of $c\Delta$
for $b$ and $c$ which are depended linearly on the radius $r$.

(ii) For a node $u\in Rec\setminus D$, grid squares that are partially in $K_u$ and partially outside.
This area is bounded by $2\sqrt{2}\mathcal{D}\Delta$. Thus, using similar methods, we can take strips of length $\Delta$,
obtaining an error bounded by
a multiplicative factor of $b'\Delta$ and additive term $c'\Delta$, for $b'$ and $c'$ which are depended linearly on $\mathcal{D}$.

Overall, the total accumulated error obtained by the modified Algorithm~\ref{alg:EDCC}, in leading terms as $\Delta \rightarrow 0$, is 
with multiplicative factor $\dfrac{a|Rec|}{2r} \sqrt{\Delta}$
and an additive term of $aMr\mathcal{D}|Rec|\sqrt{\Delta}$, 
where $a \leq const\mathcal{D}^2 \dfrac{1}{\sqrt{r}}$ (see section~\ref{subsec:pre}).

\end{proof}

The following theorem gives a bound on the running time of the modified $EDCC$ algorithm, independent of the maxima of $f$ and $g$.

\begin{theorem}
\label{thm:runtime_mult}
For a $Rec$ of area $A$ with diagonal length $\mathcal{D}$, attack of radius $r$,
a multiplicative accuracy parameter $\varepsilon>0$ and an additive accuracy parameter $\epsilon>0$, 
the total running time of the modified Algorithm~\ref{alg:EDCC} ($EDCC$) is
$O\left( \dfrac{A^{10}\mathcal{D}^{16}}{\varepsilon^8r^{12}} + \dfrac{A^{10}\mathcal{D}^{24}M^8}{\epsilon^8r^4} \right)$
where $M$ is a bound on the variation rate of $f(\cdot)$, $g(\cdot)$, and the product $f(\cdot)f(\cdot)g(\cdot,\cdot)$ over $Rec$.
\end{theorem}

\begin{proof}
As described in the proof of Theorem~\ref{thm:runtime_add},
the number of grid points in the rectangle is $A/\Delta^2$. 
Thus the running time is at most proportional to
the number of pairs of grid points, which is $O\left(A^2/\Delta^4\right)$.

From the proof of Theorem~\ref{thm:acc_mult}, we obtain that as $\Delta \rightarrow 0$
$$ \varepsilon = O\left( \frac{\mathcal{D}^2A^2}{r^{1.5}} \sqrt{\Delta} \right),$$
$$ \epsilon = O\left( M\sqrt{r}\mathcal{D}^3A\sqrt{\Delta}   \right).$$
Thus, if 
$\Delta = \Omega\left(\dfrac{\varepsilon^2r^3}{\mathcal{D}^4A^2}  \right)$,
the total running time of the algorithm is
$$O\left(A^2/\Delta^4\right) = O\left( \frac{A^{10}\mathcal{D}^{16}}{\varepsilon^8r^{12}} \right).$$
Otherwise, $\Delta = \Omega\left( \dfrac{\epsilon^2r}{M^2\mathcal{D}^6A^2} \right)$
and the total running time of the algorithm is 
$$O\left( \frac{A^{10}\mathcal{D}^{24}M^8}{\epsilon^8r^4} \right).$$
\end{proof}
\section{Find Sensitive Locations Scheme}
\label{chap:FSL}
\subsection{Sensitivity Map for Circular Attacks and Maximum Impact}
In the previous section we showed how to evaluate the damage of a cut $D$ in a specific location. 
Using Algorithm~\ref{alg:EDCC} ($EDCC$) one can approximate the TEC for a circular attack at every point,
and in particular, find an approximated worst case attack (one with the highest TEC value).  

To achieve this goal, we divide $Rec$ into squares, forming a grid. 
Then, we execute $EDCC$ algorithm from the previous section
for every grid point (squares center-points) such that it is a center-point of a circular cut of radius $r$.
This leads to a ``network sensitivity map'', i.e., for every point we have an approximation of the damage by a possible attack in that point.

The approximated worst cut is given by taking the point with the highest TEC value among all the centers of grid squares.
The actual worst case cut can be potentially located at any 
point within the grid squares whose centers' calculated TECs are approximated by the $EDCC$ algorithm.
To guarantee an attack location with TEC of at least $C-\epsilon$ where $C$ is the TEC of the actual worst cut and an accuracy parameter $\epsilon>0$, we provide algorithm $FSL$ (see pseudo-code in Algorithm \ref{alg:FSL}).

\begin{algorithm}[t]
 \caption{FindSenstiveLocations (FSL): Approximation algorithm for the network sensitivity map under a circular attack}
\label{alg:FSL}
\begin{algorithmic}[1]

\STATE For a network $N$, a cut of radius $r$, and accuracy parameter $\epsilon>0$
, apply the function $computeGrid(Rec,r,\epsilon/2)$ to find $\Delta>0$ such that 
the accuracy of Algorithm~\ref{alg:EDCC}, given by Theorem \ref{thm:acc_add} is $\epsilon/2$.
\STATE Form a grid of constant $\Delta$ (found in step 1) from $Rec$. For every grid point $p$,
apply procedure $EDCC(N,cut(p,r),\epsilon/2)$. The grid point with the highest calculated TEC is the center of the approximated worst cut.

\end{algorithmic}
\end{algorithm}

We now prove the correctness of Algorithm \ref{alg:FSL} ($FSL$). 
\begin{theorem}
\label{theorem:FSL}
For an accuracy parameter $\epsilon>0$, the attack with the highest TEC value $\tilde C$ found by the above algorithm 
satisfies $\tilde C\geq C-\epsilon$, where $C$ is the TEC value for the actual worst cut.
\end{theorem}

\begin{proof}
By Theorem \ref{thm:acc_add} for every $\epsilon'>0$ one can find a grid constant $\Delta>0$ such
that for any point $p\in Rec$ the TEC value of $cut(p,r)$ obtained by algorithm $EDCC$ is within $\epsilon'$-accuracy (additive)
from the actual TEC value of $cut(p,r)$.

For any cut located at a grid point, take a cut located at some other point within the grid square,
so it is within a distance $d<\Delta$ from the center of the square. 
The difference between the TEC for these two cases is exactly the same as in the symmetric case, where 
the functions $f$, $g$ and the grid, are shifted a distance $d$ in the other direction.
Thus, using similar arguments as in the proof of Theorem \ref{thm:acc_add}, we obtain that the difference 
is at most $\epsilon'$.
Taking $\epsilon'=\epsilon/2$ completes the proof.
\end{proof}


The following theorem determines the running time of Algorithm~\ref{theorem:FSL_runtime} for an additive accuracy parameter $\epsilon>0$.

\begin{theorem}
\label{theorem:FSL_runtime}
For a $Rec$ of area $A$ with diagonal length $\mathcal{D}$, attack of radius $r$, and an additive accuracy parameter $\epsilon>0$, 
the total running time of Algorithm~\ref{alg:FSL} ($FSL$) is
$O\left(\dfrac{A^{15}\mathcal{D}^{24}T^{12}}{\epsilon^{12}r^6} + \dfrac{A^{15}M^6}{\epsilon^6} \right)$,
where $T=\max_{u,v\in Rec}{\{f(u)f(v)g(u,v)\}}$, and $M$ is the supremum on the variation rate of $f(\cdot)f(\cdot)g(\cdot,\cdot)$ over $Rec$.
\end{theorem}
\begin{proof}
The algorithm first determines a grid constant $\Delta$ such that the accuracy of Algorithm~\ref{alg:EDCC}, given by Theorem \ref{thm:acc_add} is $\epsilon/2$. Then samples a circular cut of radius $r$ at the center point $p$ of each grid square. For each such cut, the algorithm executes $EDCC(N,cut(p,r),\epsilon/2)$ in $O(A^2/\Delta^4)$ time. The grid has $O\left(A/\Delta^{2}\right)$ points.
Thus the total running time is at most $O\left(A^3/\Delta^6\right)$.
 
Now, from the proof of Theorem \ref{thm:runtime_add} we obtain that
when $M=O\left(\dfrac{\mathcal{D}^4T^2}{\epsilon r}\right)$, 
$$\Delta = \Omega\left(\frac{\epsilon^2 r}{\mathcal{D}^4A^2T^2}\right),$$
and thus, the total running time of the algorithm is
$$O(A^3/\Delta^6)=O\left(\frac{A^{15}\mathcal{D}^{24}T^{12}}{\epsilon^{12}r^6}\right).$$
Otherwise, if $M=\omega\left(\dfrac{\mathcal{D}^4T^2}{\epsilon r}\right)$, 
we obtain that $\Delta=\Omega\left(\dfrac{\epsilon}{MA^2}\right)$,
and the running time of the algorithm is $O\left(\dfrac{A^{15}M^6}{\epsilon^6}\right)$.
\end{proof}

Algorithm~\ref{alg:FSL} can be modified to give a multiplicative approximation which is independent of the maxima of $f$ and $g$ (over $Rec$),
by using the modified function $computeGrid(Rec,r,\epsilon/2, \varepsilon/2)$ and the modified $EDCC$ algorithm 
(as described in section~\ref{subsec:multiplicative}).
The correctness of this algorithm is obtained similarly as in Theorem~\ref{theorem:FSL}.
For the running time we provide the following theorem.

\begin{theorem}
\label{theorem:FSL_runtime_mult}
For a $Rec$ of area $A$ with diagonal length $\mathcal{D}$, attack of radius $r$,
a multiplicative accuracy parameter $\varepsilon>0$ and an additive accuracy parameter $\epsilon>0$, 
the total running time of the modified $FSL$ algorithm is
$O\left( \dfrac{A^{15}\mathcal{D}^{36}M^{12}}{\epsilon^{12}r^6} + \dfrac{A^{15}\mathcal{D}^{24}}{\varepsilon^{12}r^{18}} \right)$,
where $M$ is a bound on the variation rate of $f(\cdot)$, $g(\cdot)$, and the product $f(\cdot)f(\cdot)g(\cdot,\cdot)$ over $Rec$. 
\end{theorem}
\begin{proof}
As described in Theorem~\ref{theorem:FSL_runtime}, the total running time is at most $O\left(A^3/\Delta^6\right)$.
Now, for the modified $EDCC$ algorithm described in section~\ref{subsec:multiplicative}, 
we obtain by the proof of Theorem~\ref{thm:runtime_mult} that the grid constant $\Delta$ satisfies
$\Delta = \Omega\left(\dfrac{\varepsilon^2r^3}{\mathcal{D}^4A^2}  \right)$ or (the asymptotic minimum)
$\Delta = \Omega\left( \dfrac{\epsilon^2r}{M^2\mathcal{D}^6A^2} \right)$

Thus, the total running time of the modified $FSL$ algorithm is
$$O\left( \frac{A^{15}\mathcal{D}^{36}M^{12}}{\epsilon^{12}r^6} + \frac{A^{15}\mathcal{D}^{24}}{\varepsilon^{12}r^{18}} \right).$$

\end{proof}

\subsection{Random Attacks}
\label{sec:random}
The impact of a \emph{random circular cut} to \emph{deterministic} network was recently studied (see section~\ref{sec:related}).
The random location of a disaster can model failure resulting from a natural 
disaster such as a hurricane or collateral (non-targeted) damage in an EMP attack.
An interesting question is to evaluate the expected impact of a random circular cut to our \emph{stochastic} network model.

Algorithm~\ref{alg:EDCC} ($EDCC$) gives an approximation for the expected damage caused by a circular cut which is located at a specific point.
We can use it to develop an algorithm for evaluating the damage caused by a random circular cut, using a similar concept
as in Algorithm~\ref{alg:FSL} ($FSL$). 

For a random circular cut distributed \emph{uniformly} over $Rec$,
an additive approximation of the expected damage caused by such a random cut to our stochastic network model is given by

\begin{equation}
\label{eq:random}
\frac{1}{|Rec|}\iint\limits_{Rec} EDCC(N,cut(p,r),\epsilon')\, dp. 
\end{equation}
For a given additive accuracy parameter $\epsilon > 0$,
the idea is to evaluate the above equation numerically over a grid, with a grid constant $\Delta$ small enough, and an appropriate $\epsilon'$,
such that the total additive error will be at most $\epsilon$. This is provided in Algorithm $RCCE$ (see pseudo-code in Algorithm~\ref{alg:RCCE}). 

\begin{algorithm}[t]
 \caption{RandomCircularCutEvaluation (RCCE): Approximation algorithm for evaluating the expected damage caused by a random circular attack.}
\label{alg:RCCE}
\begin{algorithmic}[1]

\STATE For a random network $N$, a random cut of radius $r$, distributed uniformly over $Rec$,
and an additive accuracy parameter $\epsilon>0$, 
apply the function $computeGrid(Rec,r,\epsilon/2)$ to find $\Delta>0$ such that 
the accuracy of Algorithm~\ref{alg:EDCC}, given by Theorem \ref{thm:acc_add} is $\epsilon/2$.
\STATE Form a grid of constant $\Delta$ (found in step 1) from $Rec$. Denote by $Grid$ the set\\ of grid points (square centers).
\RETURN  $\frac{1}{|Rec|}\sum_{p\in Grid} EDCC(N,cut(p,r),\epsilon/2)$

\end{algorithmic}
\end{algorithm}

Similarly as in Algorithm~\ref{alg:FSL} ($FSL$), by the proof of Theorem~\ref{theorem:FSL}, we obtain that
the grid constant $\Delta$, chosen in Algorithm~\ref{alg:RCCE} guarantees that 
for two cuts of radius $r$, located at points $u$ and $v$ within euclidean distance at most $\Delta$ from each other, 
the difference in the TEC value for these two cuts is at most $\epsilon$.
Thus, when evaluating (\ref{eq:random}) numerically in Algorithm~\ref{alg:RCCE},
the total error is at most $\dfrac{1}{|Rec|} \cdot |Rec| \cdot \epsilon = \epsilon$.
Thus, the correctness of Algorithm~\ref{alg:RCCE} is obtained similarly as for the $FSL$ algorithm, given by Theorem~\ref{theorem:FSL}.
The running time of Algorithm~\ref{alg:RCCE} is also the same as for the $FSL$ algorithm, given by Theorem~\ref{theorem:FSL_runtime}.

Using the multiplicative approximation with the modified $EDCC$ algorithm and the modified function $computeGrid(Rec,r,\epsilon/2, \varepsilon/2)$ 
(as described in section~\ref{subsec:multiplicative}),
a similar result with running time as in Theorem~\ref{theorem:FSL_runtime_mult} can be obtained for a combined multiplicative and additive approximation of the worst-case cut, independent of the maxima of $f$ and $g$ over $Rec$.

A random circular cut can be modeled more generally when its location is distributed
with some distribution function with density $\psi$ over $Rec$.
Similarly to the uniform case,
an approximation of the expected damage caused by such a random cut to our stochastic network model is given by
\begin{equation}
\label{eq:random_gen}
\iint\limits_{Rec} \psi(p)EDCC(N,cut(p,r),\epsilon')\, dp.
\end{equation}
For a random circular cut distributed with density function $\psi(p)$ over $Rec$,
such that $\psi$ is a function of \emph{bounded variation} over $Rec$,
A \emph{multiplicative} approximation algorithm can be obtained by applying
similar techniques as in Algorithm~\ref{alg:RCCE} and Algorithm~\ref{alg:FSL}, choosing the grid constant $\Delta$
small enough, such that multiplicative error factor will be at most $\epsilon$.
The approximation and running time in this case depends also on the \emph{supremum} on the variation rate of $\psi$ over $Rec$
(which was zero in the uniform case),
due its contribution to the (multiplicative) accumulated error   
when evaluating numerically the integral (\ref{eq:random_gen}) over the grid.

\subsection{Simulations and Numerical Results}
\label{sec:numerical_results}

\begin{figure}[t]
\centering
	  \includegraphics[width=0.9\textwidth]{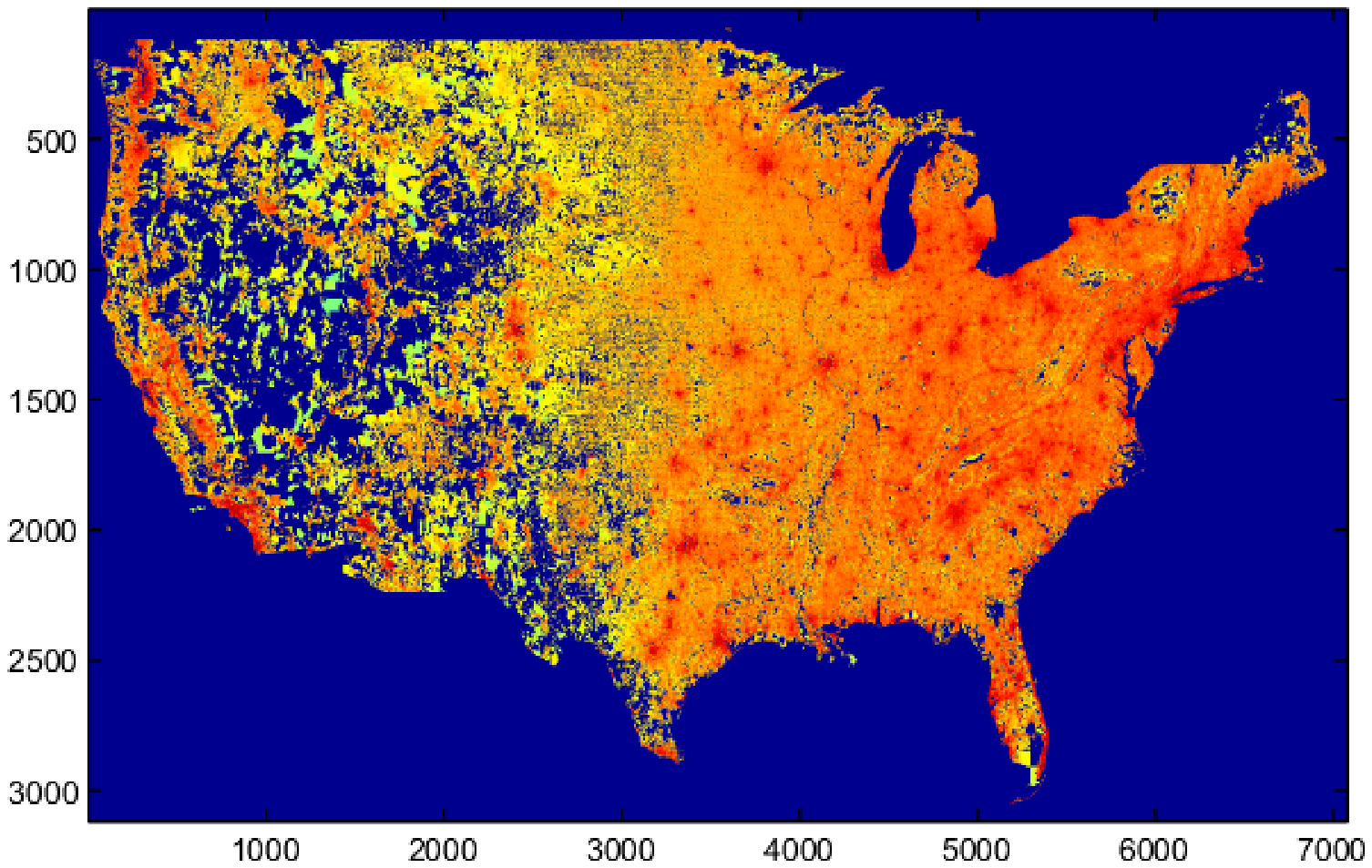}
 \caption{ Color map of the the USA population density in logarithmic scale. Data is taken from \cite{columbia}. }
       \label{fig:pop_density}
\end{figure}

In order to test our algorithms, we estimate the expected impact of 
circular cuts on communication networks in the USA based on a population density map.
We implemented the algorithm 
as a C program. Data for the population density of the USA taken from 
\cite{columbia} was taken as the intensity function $f(u)$. The supplied data 
is the geographic density of the USA population at a resolution of $30"$, which is approximately
equivalent to 0.9km. The matrix given was of  dimensions $3120\times7080$. See Fig. \ref{fig:pop_density}.

In order to achieve faster running times, the data was averaged over 
$30\times30$ blocks, to give a resolution of approximately 27km. This gave a 
matrix of  dimensions $104\times 236$. The algorithm was then run over 
this intensity function. The function $y(u,v)$ was taken to
be $1/\mathrm{dist}(u,v)$ , where $\mathrm{dist}$ is the Euclidean distance
between the points, based on observations that the lengths of
physical Internet connections follow this distribution \cite{ref:BA99}. 
The capacity probability function $h$ was taken to be constant, independent of the distance, reflecting
an assumption of standard equipment. Each run took around 24 hours on a
standard Intel CPU computer.

Results for different cut radii are given in Fig. \ref{fig:rad5}, Fig. \ref{fig:rad8} and Fig. \ref{fig:rad10}.
As can be seen , the most harmful cut would be around NYC, as expected, where 
for the larger cut radius, a cut between the east and west coast, effectively disconnecting 
California from the north-east, would also be a worst case scenario. Lower, but still 
apparent maxima are observed in the Los-Angeles, Seattle and Chicago areas.

Comparing the results to the results obtained for the fiber  backbone in \cite{NZCM.INFOCOM09} (see Fig.~\ref{fig:fiber})
it can be seen that some similarities and some dissimilarities exist.  While the NYC maxima
is apparent in all measures,  California seems to be missing from the maxima in 
\cite{NZCM.INFOCOM09}, probably reflecting the low density of fibers in that area in the map 
studied  in~\cite{NZCM.INFOCOM09}.
\begin{figure}[t]
\centering
 \includegraphics[width=0.9\textwidth]{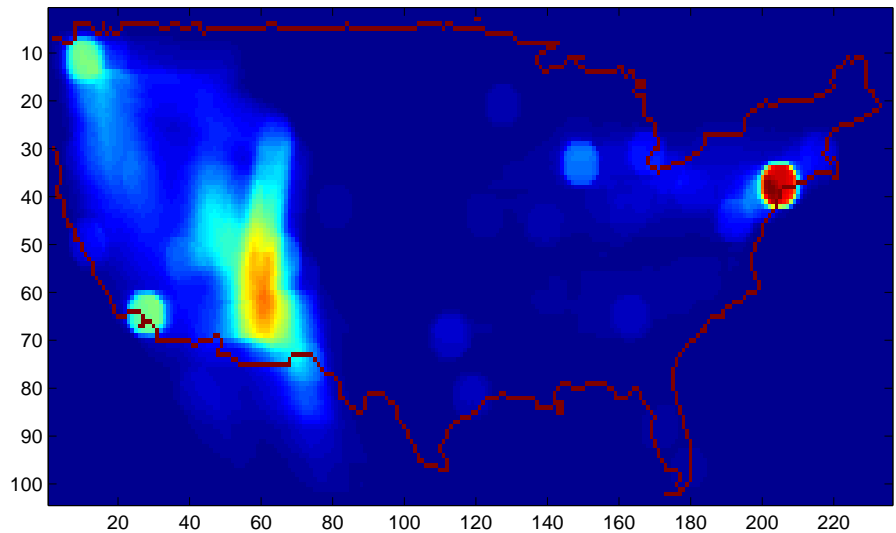}
 \caption{Color map of the centers of circular cuts with radius $r=5$  (approximately 130km). Red is most harmful.}
       \label{fig:rad5}
\end{figure}

\begin{figure}[H]
\centering
        \includegraphics[width=0.9\textwidth]{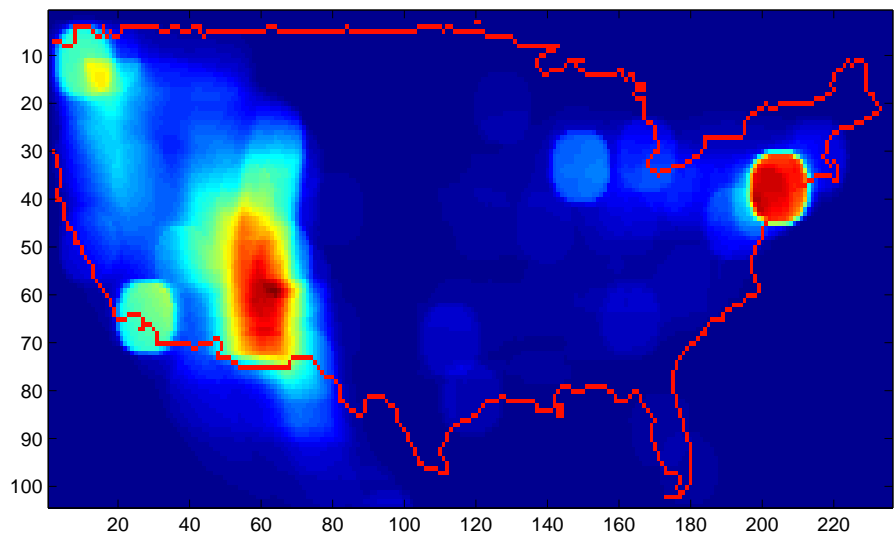}
 \caption{Color map of the centers of circular cuts with radius $r=8$ (approximately 208km). Red is most harmful.}
       \label{fig:rad8}
\end{figure}
As many fibers may exist that are not represented
on the map in \cite{NZCM.INFOCOM09}, it may be reasonable to assume that a 
cut around California would have a more significant effect than reflected there.
On the other hand, several worst case cuts in Texas, and especially in Florida are apparent in 
 \cite{NZCM.INFOCOM09} and are missing in the current simulation results. It seems
that the effect of the narrow land bridge in Florida makes cuts there especially 
harmful, whereas our simulation assumes links are straight lines, which will make
links to both east and west coast pass through the ocean, thus making cuts 
less harmful.

As a full map of communication lines is not available, it is still unclear how good of an
approximation the results here supply. However, the tool can be used in conjunction
with more complicated modeling assumptions, including topographic features and
economic considerations to give more exact results.

\begin{figure}[t]
\centering
        \includegraphics[width=0.9\textwidth]{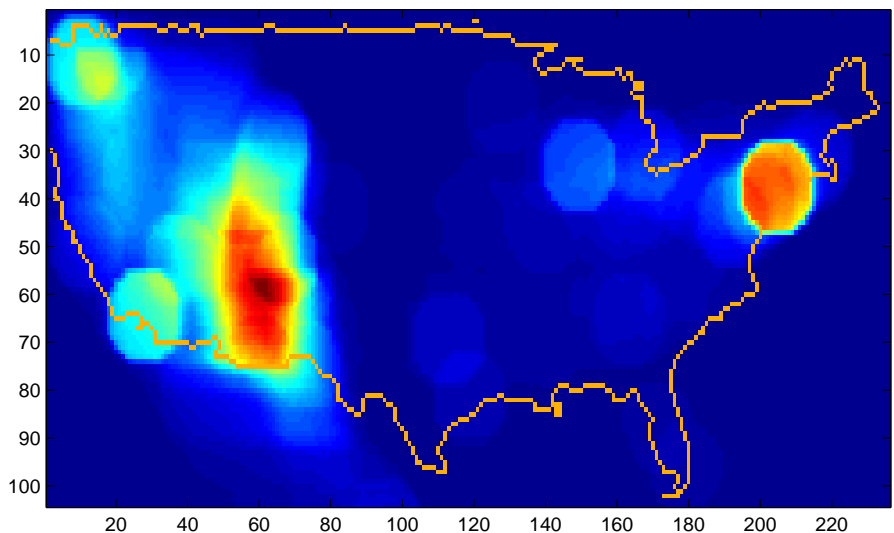}
 \caption{Color map of the centers of circular cuts with radius $r=10$ (approximately 260km). Red is most harmful.}
       \label{fig:rad10}
\end{figure}

\section{Conclusions and Future Work}

\paragraph{Conclusions}
Motivated by applications in the area of network robustness
and survivability, we focused on the problem
of geographically correlated network failures. 
Namely, on assessing the impact of such failures to the network.
While previous works in this emerging field focused mainly on deterministic networks (as described in section~\ref{sec:related}),
we studied the problem under non-deterministic network models.

We proposed a method to stochastically model the geographical (physical) layout of communication networks,
i.e. the geographical locations of nodes and links, as well as the capacity of the links.
This applies to both the case where the network is derived from a random
model, or to the case where the physical topology of the network is unknown (or partially known) to the attacker,
who possesses some statistical information (e.g., population density, topography, economy) 
about the geographical area of the network,
or the probability of having link between various locations.

Using tools from geometric probability and numerical analysis, we developed approximation algorithms
for finding the damage caused by physical disasters (modeled by circular cuts)
at different points and to approximate the location
and damage of the worst-case cuts for this model. 
We also provided an algorithm to assess
the impact of a random circular disaster (i.e. non-targeted) to our random network model, 
motivated by modeling a failure resulting from a natural 
disaster such as a hurricane or collateral (non-targeted) damage in an EMP attack.
We proved the correctness of our schemes and 
mentioned the trade between running time and accuracy, for both additive and multiplicative error terms.

%
In order to test the applicability of our model and algorithms to real-world scenarios,
we applied our algorithms to approximate the damage caused by cuts in different 
locations to communications networks in the USA, where the network's geographical layout
was modeled probabilistically, relying on demographic information (i.e. population density) only.
We found a strong correlation between locations of cuts that cause high relative damage to the population density distribution
over the network's region.
Some of the results agree with the exact results obtained before 
about the fiber backbone of a major network provider in the USA and some do not (as described in section~\ref{sec:numerical_results}).
It is yet to be determined if taking 
into account more links would lead to results closer to our scheme's prediction
or whether the results are fundamentally different due to an inexact link model.

Our results imply that some information on the network sensitivity and 
vulnerabilities can be deduced from the population alone, with no information
on any physical links and nodes.
Thus, our schemes allows
to examine how valuable is public information
(such as demographic, topographic and economic information)
to an attacker's destruction assessment capabilities, and examine
the affect of hiding the actual physical location of the fibers on the attack strategy.
Thereby, the schemes can be used as a tool for policy makers and engineers to
design more robust networks by placing links along paths that avoid 
areas of high damage cuts, or identifying locations which require
additional protection efforts (e.g., equipment shielding).

\paragraph{Future Work} 
The discussion about finding vulnerable geographic locations to physical attacks naturally leads
to the question of robust network design in the face of geographical failures.
Several questions are proposed, one is to investigate the effect of
adding minimal infrastructure (e.g., lighting-up dark fibers) on network resilience, 
and determine how to use low-cost shielding for existing
components to mitigate large-scale physical attacks.
Another question is on designing the network's physical topology 
under some demand constraints (e.g., nodes that should be located within a specific region, capacity and flow demands) 
such that the damage by a large-scale physical attack is minimized.

Another related research direction is to
develop a framework for attack and defense strategies for opponents having no knowledge of each other's strategy.
Using a game-theoretic approach, study a two player zero sum game where one player (the defender)
attempts to design a network as resilient to physical attacks as possible under some demand constraints,
while the other player (the attacker) picks a location for the cut,
without having complete knowledge about the actual network's physical structure.




     
\bibliographystyle{IEEEtranS}
\bibliography{icc-paper}

\end{document}